\documentclass[journal,letterpaper,twocolumn,final,10pt]{IEEEtran}
\usepackage{color}
\usepackage[final]{graphicx}
\usepackage{amssymb}
\usepackage{amsmath}
\usepackage{amsthm}
\usepackage{dsfont}
\usepackage{cite}
\usepackage{algorithmic}
\usepackage{algorithm}
\usepackage{subfigure}

\newcommand{\eit}{\end{itemize}}

\newcommand{\ben}{\begin{enumerate}}
\newcommand{\een}{\end{enumerate}}

\newcommand{\bdesc}{\begin{description}}
\newcommand{\edesc}{\end{description}}

\newcommand{\bea}{\begin{array}}
\newcommand{\eea}{\end{array}}

\newcommand{\beqa}{\begin{eqnarray}}
\newcommand{\eeqa}{\end{eqnarray}}

\newcommand{\ds}{\displaystyle}

\newcommand{\Comment}[1]{}

\newtheorem{prop}{Proposition}

\def\N{{\mathds N}}

\def\R{{\mathds R}}

\def\C{{\mathds C}}

\newcommand{\be}{\begin{equation}}
\newcommand{\ee}{\end{equation}}

\newcommand{\bzero}{{\mbox{\boldmath $0$}}}

\newcommand{\bff}{{\mbox{\boldmath $f$}}}
\newcommand{\bn}{{\mbox{\boldmath $n$}}}
\newcommand{\bm}{{\mbox{\boldmath $m$}}}

\newcommand{\bor}{{\mbox{\boldmath $r$}}}

\newcommand{\bs}{{\mbox{\boldmath $s$}}}

\newcommand{\bv}{{\mbox{\boldmath $v$}}}
\newcommand{\bx}{{\mbox{\boldmath $x$}}}

\newcommand{\bz}{{\mbox{\boldmath $z$}}}

\newcommand{\bA}{{\mbox{\boldmath $A$}}}

\newcommand{\bF}{{\mbox{\boldmath $F$}}}

\newcommand{\bI}{{\mbox{\boldmath $I$}}}

\newcommand{\bM}{{\mbox{\boldmath $M$}}}

\newcommand{\bS}{{\mbox{\boldmath $S$}}}

\newcommand{\bV}{{\mbox{\boldmath $V$}}}

\newcommand{\bZ}{{\mbox{\boldmath $Z$}}}

\newcommand{\balpha}{{\mbox{\boldmath $\alpha$}}}

\newcommand{\btheta}{{\mbox{\boldmath $\theta$}}}

\newcommand{\tRAO}{t_{\mbox{\tiny RAO}}}

\newcommand{\tTSGLRT}{t_{\mbox{\tiny TS-GLRT}}}

\newcommand{\Halpha}{\widehat{\alpha}}

\newcommand{\dmax}{\begin{displaystyle}\max\end{displaystyle}}
\newcommand{\dmin}{\begin{displaystyle}\min\end{displaystyle}}

\newcommand{\test}{\mbox{$
\begin{array}{c}
\stackrel{ \stackrel{\textstyle H_1}{\textstyle >} }{
\stackrel{\textstyle <}{\textstyle H_0} }
\end{array}
$}}

\title{Adaptive Detection of Point-like Targets in Spectrally Symmetric Interference}

\vspace{0.1cm}

\author{A. De Maio, \IEEEmembership{Fellow, IEEE}, D. Orlando, \IEEEmembership{Senior Member, IEEE}, C. Hao, \IEEEmembership{Member, IEEE}, 
and G. Foglia, \IEEEmembership{Member, IEEE,}
\thanks{A. De Maio is with the Dipartimento di Ingegneria Elettrica e delle Tecnologie dell'Informazione, Universit\`a degli Studi di Napoli ``Federico II'',
via Claudio 21, I-80125 Napoli, Italy. E-mail: {\tt ademaio@unina.it}.}
\thanks{D. Orlando is with Engineering Faculty of Universit\`a degli Studi ``Niccol\`o Cusano'', via Don Carlo Gnocchi 3, 00166 Roma, Italy. 
E-mail: {\tt danilo.orlando@unicusano.it}.}
\thanks{C. Hao is with the State Key Laboratory of Information Technology for Autonomous Underwater vehicles, 
Chinese Academy of  Sciences, 100190 Beijing, China. E-mail: {\tt haochengp@mail.ioa.ac.cn}.}
\thanks{G. Foglia is with ELETTRONICA S.p.A., Via Tiburtina Valeria Km 13,700, 00131 Roma, Italy. 
E-mail: {\tt goffredo.foglia@gmail.com}.}
\thanks{This work was supported in part by the National Natural Science Foundation of China under Grant Nos. 61172166 and 61571434.}
}

\begin{document}

\maketitle

\begin{abstract}
We address adaptive radar detection of targets embedded in ground clutter dominated environments
characterized by a symmetrically structured power spectral density. 
At the design stage, we leverage on the spectrum symmetry for the interference to
come up with decision schemes capable of capitalizing the a-priori information on the covariance structure. 
To this end, we prove that the detection problem at hand can be formulated in terms of real variables
and, then, we apply design procedures relying on the GLRT, the Rao test, and the Wald test.
Specifically, the estimates of the unknown parameters under the target presence hypothesis 
are obtained through an iterative optimization algorithm whose convergence and quality guarantee is thoroughly proved. 
The performance analysis, both on simulated and on real radar data, confirms the superiority of the considered architectures 
over their conventional counterparts which do not take advantage of the clutter spectral symmetry.
\end{abstract}


\section{Introduction}

\IEEEPARstart{I}{n the last years}, radar community undertook different routes towards the design of adaptive detection schemes.
The most common design criteria as the Generalized Likelihood Ratio Test (GLRT), the Rao test, and the Wald test have been
exploited in conjunction with specific conditions on the interference affecting the target echoes usually arising in 
some  operating scenarios.

The seminal approach by Kelly et al. \cite{Kelly86,Kelly-Nitzberg,Kelly-TR} did not foresee any
additional assumption on the spectral properties of the interference except for
the circular symmetry. The authors suppose that a set of secondary data, free of signal components and sharing the same
spectral properties of the data under test (primary data), is available to estimate the interference covariance matrix. 
However, this scenario (also known as homogeneous scenario) dictates a constraint on the the number, $K$ say, of secondary data.
More precisely, conventional decision schemes require the inversion of 
the sample covariance matrix. To this end, it is required that 
$K$ has to be greater than or equal to the dimension of the data vectors, $N$ say. Additionally, 
detection performances are strongly affected by 
the estimation quality of the interference covariance matrix \cite{Kelly86}, which relates to $K$.
A lower bound on $K$, which ensures good detection performances, is $2N$, i.e., $K\geq 2N$.
When this condition is not fulfilled due, for instance, to heterogeneity 
between primary and secondary data, severe performance degradations are experienced 
\cite{Melvin-2000,gini1,gini2}.
As a matter of fact, secondary data are often contaminated by power variations over range, clutter discretes, and other outliers, which
drastically reduce number of homogeneous secondary data.
Adaptive detection of signals buried in interference environments for which the secondary data volume is not large is 
referred to as {\em sample-starved problem} \cite{Yuri00,Yuri01}.

Strategies conceived to cope with such situations exhibit a common denominator that consists in incorporating the available {\em a priori}
information into the detector design (knowledge-aided paradigm). For instance, in \cite{DeMaioFoglia00}, the authors show that significant
performance improvements can be achieved exploiting the available information about the surrounding environment.
In particular, they propose algorithms which use the information provided by a geographic information system
in order to properly select secondary data. Another example is provided in \cite{DeMaioFoglia01}, where
the Bayesian approach is employed assuming that the unknown covariance matrix of the interference
obeys a suitable distribution. Under this hypothesis, two GLRT-based detectors are derived and the performance analysis
on real data reveals the superiority of the proposed detectors with respect to their non-Bayesian counterparts when the training set is small.
The Bayesian framework can be also used together with the structural information on the interference 
covariance matrix \cite{Aubry} as shown in \cite{Hongbin}, where the disturbance is modeled as a multi-channel auto-regressive process
with a random cross-channel covariance matrix (see also \cite{Hongbin2,Wiesel}). 

In radar applications, where systems are equipped with array of sensors, 
structural information about the interference covariance matrix arises from the exploitation 
of specific class of geometries.
For instance, in the case of a symmetrically spaced linear array or a system transmitting symmetrically spaced pulse trains, 
collected data could be statistically symmetric in forward and reverse directions.
This results into an interference covariance matrix which shares a so-called ``doubly'' symmetric form, i.e., Hermitian about its principal 
diagonal and persymmetric about its cross diagonal \cite{Nitzberg-1980}.
The mentioned special structure is also induced for the steering vector and allows to achieve interesting processing gains \cite{VanTrees4}.
It is important to highlight that the persymmetric structure is not limited to linear 
arrays but it can be found in different geometries such as standard rectangular arrays, 
uniform cylindrical arrays (with an even number of elements), 
and some standard exagonal arrays \cite{VanTrees4}.
Several approaches, relying on the  persymmetry, have been developed to achieve improved 
detection performances in training-limited scenarios; just to give some examples, see \cite{Cai1992,hongbinPersymmetric,Hongbin1,Pascal,Pascal2,Casillo-2007,PersymmetricRao,GuerciPersymmetry}.

Another source of a priori information, which can be exploited in the design of adaptive algorithms, 
is the possible symmetry  in the clutter spectral characteristics. In fact, it is well-known that  {\em ground clutter}
observed by a stationary monostatic radar often exhibits a symmetric Power Spectral Density (PSD) centered around the zero-Doppler frequency and  whose integral (clutter power) depends on the type of illuminated background \cite{Klemm}. This property has been corroborated by diverse statistical analyses on experimentally 
measured data \cite{Billingsley00,Billingsley01} and implies that clutter autocorrelation function is real-valued and even.
It represents an important structure which reduces the number of  nuisance parameters to estimate and can be exploited at the design stage.  
Specifically, collected data are organized into vectors which, from a statistical perspective, 
are modeled in terms of circularly symmetric complex Gaussian vectors. Now, if the clutter autocorrelation function is real, then
the resulting covariance matrix is real. Each complex vector is thus statistically equivalent to a pair
of independent real Gaussian vectors and the original detection problem can be transferred from the complex domain to the real domain. As a result the number of secondary data is increased by a factor $2$.

Following the above guideline, we focus on ground clutter dominated environments and design 
four adaptive decision schemes which leverage on the symmetric PSD structure for the interference. 
We first transform the problem from the complex domain to the real domain and then solve the new hypothesis test resorting to
design procedures relying on the GLRT, the Rao test, and the Wald test.
It is worth observing that the mathematical derivation of the plain GLRT and the Wald test
for the problem at hand is a formidable task (at least to the best of authors' knowledge). 
For this reason, we exploit {\em ad-hoc} suboptimum procedures (but with a quality guarantee), 
which are suitable modifications of previous criteria, to devise
four adaptive decision schemes.
More precisely, the first is obtained by means of the well-known two-step GLRT-based design procedure \cite{Kelly-Nitzberg}, whereas
the second, which is asymptotically equivalent to the plain GLRT, is devised according to the following rationale
\begin{enumerate}
\item the plain GLRT is evaluated assuming that target amplitudes are perfectly known;
\item target amplitudes are replaced by suitable estimates provided by a newly 
proposed iterative estimation algorithm exploiting alternating (cyclic) optimization and sharing quality guarantee.
\end{enumerate}
The last two architectures are devised employing the Rao test design criterion and an ad-hoc modification of the Wald test which
exploits the amplitude estimates provided by the aforementioned iterative algorithm.
The performance analysis confirms the superiority of the considered architectures 
over their conventional counterparts which do not capitalize on the real and even PSD of the clutter.

The remainder of this paper is organized as follows. Section II addresses the problem formulation while 
Section III deals with the design of the detectors. Section IV provides illustrative examples. 
Some concluding remarks and hints for future work are given in Section V. 
Finally, the appendices contain analytical derivations of the results presented in the previous sections.

\subsection{Notation}
In the sequel, vectors and matrices are denoted by boldface lower-case and upper-case letters, respectively. 
Symbols $\det(\cdot)$ and $\mbox{Tr}(\cdot)$
denote the determinant and the trace of a square matrix, respectively. If $A$ and $B$ are scalars, then $A\times B$ is the usual product of scalars;
on the other hand, if $A$ and $B$ are generic sets, $A \times B$ denotes the Cartesian product of sets.
The imaginary unit is $j$, i.e., $\sqrt{-1}=j$.
The $(i,j)$-entry of a generic matrix $\bA$ is denoted by $\{\bA\}_{i,j}$.
Symbol $\bI_N$ represents
the $(N\times N)$-dimensional identity matrix, while $\bzero$ is the null vector or matrix of proper dimensions. The Euclidean norm of a vector is denoted
by $\|\cdot\|$. As to the numerical sets, $\R$ is the set of real numbers,
$\R^{N\times M}$ is the set of $(N\times M)$-dimensional real matrices (or vectors if $M=1$),
$\C$ is the set of complex numbers, and $\C^{N\times M}$ is the set of $(N\times M)$-dimensional complex matrices (or vectors if $M=1$).
Symbol $\bS_{++}^{N}$ is used to represent the set of $N\times N$ positive definite symmetric matrices.
The real and imaginary parts of a complex vector or scalar are denoted by $\Re(\cdot)$ and $\Im(\cdot)$, respectively.
Symbols $(\cdot)^*$, $(\cdot)^T$, and $(\cdot)^\dag$ stand for complex conjugate, transpose, and conjugate transpose, respectively.
The acronym iid means independent and identically distributed while the symbol $E[\cdot]$ denotes statistical expectation.
Finally, $A \propto B$ means that $A$ is proportional to $B$.
% If $\bx\in\R^N$ ($\in\C^{N\times 1}$) is
% distributed according to the multivariate Gaussian (the circular complex normal) distribution with mean $\bm\in\R^{N\times 1}$ ($\in\C^{N\times 1}$)
% and covariance matrix $\bM\in\R^{N\times N}$ ($\in\C^{N\times N}$), we write $\bx\sim\cN_N(\bm, \bM)$ ($\sim\cC\cN_N(\bm, \bM)$). 
% Finally, $x\sim\chi^2_{K}$ means that
% $x$ obeys the central chi-square distribution with $K$ degrees of freedom.

\section{Problem Formulation}
\label{Sec:ProblemFormulation}
In this section, we introduce the detection problem at hand and show that, under the assumption of a symmetric spectrum for the interference,
it is equivalent to another decision problem dealing with real vectors and matrices. To this end, let us begin by formulating
the initial problem in terms of a binary hypothesis test.
Specifically, we assume that the considered sensing systems acquires data from $N\geq 2$ channels which can be spatial and/or temporal.
The echoes from the cell under test are properly pre-processed, namely, the received signals are downconverted to baseband or
an intermediate frequency; then, they are sampled and organized to form a $N$-dimensional vector, $\bor$ say.
We want to test whether or not $\bor$ contains useful target echoes assuming the presence of $K\geq N/2$ secondary data.
Summarizing, we can write this decision problem as follows
% \be
% \left\{
% \begin{array}{ll}
% H_0:\left\{\begin{array}{ll}
% \bor = \bn, \phantom{+\alpha \bs} &  \\
% \bor_k = \bn_k, & k=1, \ldots, K,
% \end{array}\right.
% \\
% \\
% H_1:\left\{\begin{array}{ll}
% \bor = \alpha \bv +  \bn, &  \\
% \bor_k = \bn_k, & k=1, \ldots, K,
% \end{array}\right.
% \end{array}
% \right.
% \label{eqn:hypothesistest00}
% \ee
\be
\left\{
\begin{array}{lll}
H_0: \bor = \bn, & \bor_k = \bn_k, & k=1, \ldots, K,
\\
H_1: \bor = \alpha \bv +  \bn, & \bor_k = \bn_k, & k=1, \ldots, K,
\end{array}
\right.
\label{eqn:hypothesistest00}
\ee
where
\begin{itemize}
\item $\bv=\bv_1+j\bv_2 \in\C^{N\times 1}$ with $\|\bv\|=1$, $\bv_1=\Re\{\bv\}$, and $\bv_2=\Im\{\bv\}$ is the nominal steering vector;
\item $\alpha=\alpha_1+j\alpha_2\in\C$ with $\alpha_1=\Re\{\alpha\}$ and $\alpha_2=\Im\{\alpha\}$ represents the target response
which is modeled in terms of an unknown deterministic factor accounting for target reflectivity and channel propagation effects;
\item $\bn=\bn_1+j\bn_2\in \C^{N \times 1}$ and $\bn_k=\bn_{1k}+j\bn_{2k} \in \C^{N \times 1}$,  $k=1, \ldots, K$, with $\bn_1=\Re\{\bn\}$, $\bn_2=\Im\{\bn\}$,
$\bn_{1k}=\Re\{\bn_k\}$, and $\bn_{2k}=\Im\{\bn_k\}$, are iid circular complex normal random vectors with zero mean
and unknown positive definite covariance matrix $\bM_0\in\bS_{++}^{N}$; it is important to observe here that, since the interference 
shares zero mean and exhibits a PSD symmetric with respect to the zero 
frequency\footnote{Observe that if the PSD of the clutter is a real and even function, then, 
due to the time/frequency duality of Fourier Transform and to the Wiener-Khintchine Theorem, the autocorrelation function
of the clutter is real and even. As a consequence, the clutter covariance matrix belongs to $\bS_{++}^N$.}
(i.e., the PSD is an even function), the covariance of the interfering signals is real and even.
\end{itemize}
Now, recall that a zero-mean complex Gaussian vector $\bx=\bx_1+j\bx_2\in\C^{N \times 1}$, $\bx_1=\Re\{\bx\}$ and 
$\bx_2=\Im\{\bx\}$, is said to be circular complex normal \cite{BOR-Morgan} if $E[\bx_1\bx_1^T]=E[\bx_2\bx_2^T]$,
$E[\bx_1\bx_2^T]=-E[\bx_2\bx_1^T]$. Under the above assumption, the covariance matrix of $\bx$ can be written as
\be
E[\bx\bx^\dag]=2(E[\bx_1\bx_1^T]-jE[\bx_1\bx_2^T])\in\C^{N\times N}.
\ee
In (\ref{eqn:hypothesistest00}), we have modeled the disturbance in terms of circular complex normal random vectors with zero mean and
real covariance matrix, which, in turn, implies that the cross-covariances between the real and imaginary parts of $\bn$ and $\bn_k$,
$k=1,\dots,K$, are zero. Thus, we can claim that $\bn_1$, $\bn_2$ and $\bn_{1k}$, $\bn_{2k}$, $k=1,\dots,K$, are iid real Gaussian vectors
with zero mean and covariance matrix $\bM=\frac{1}{2}\bM_0\in\R^{N\times N}$.
As a consequence, we can recast problem (\ref{eqn:hypothesistest00}) into the equivalent form
\be
\left\{
\begin{array}{ll}
H_0:\left\{\begin{array}{ll}
\bz_1 = \bn_1, \ \bz_2 = \bn_2,\phantom{+\alpha \bs} &  \\
\bz_{1k} = \bn_{1k}, \ \bz_{2k} = \bn_{2k}, & \  \ \ k=1, \ldots, K,
\end{array}\right.
\\
\\
H_1:\left\{\begin{array}{ll}
\bz_1 = (\alpha_1 \bv_1-\alpha_2\bv_2) +  \bn_1, & \\
\bz_2 = (\alpha_1 \bv_2 + \alpha_2\bv_1) +  \bn_2, &  \\
\bz_{1k} = \bn_{1k}, \ \bz_{2k} = \bn_{2k}, &  k=1, \ldots, K.
\end{array}\right.
\end{array}
\right.
\label{eqn:hypothesistest01}
\ee
The above problem is formally equivalent to (\ref{eqn:hypothesistest00}). As a matter of fact, for the latter problem, 
the relevant parameter to decide for the presence of a target is $\alpha$, or, equivalently, the pair $(\alpha_1,\alpha_2)$. After transformation leading to (\ref{eqn:hypothesistest01}), the formal structure of the decision problem is again
$H_0: \ (\alpha_1,\alpha_2)=(0,0)$, $H_1: \ (\alpha_1,\alpha_2)\neq (0,0)$.

In the next section, we focus on problem (\ref{eqn:hypothesistest01}) and devise adaptive decision schemes based upon the
GLRT, the Rao, and the Wald test design criteria.

\section{Detector Designs}
In this section, four different decision rules are proposed. The first two rely on suitable modifications of the GLRT design criterion.
In particular, we consider the so-called two-step GLRT which consists in evaluating the GLRT of the cell under test assuming that
$\bM$ is known and then replacing it with a proper estimate. On the other hand, the second architecture is conceived exploiting a recursive
estimation strategy of the target response within the GLRT framework (this point is better explained in what follows).
The third decision scheme comes from the application of the Rao test design criterion to the problem at hand.
Finally, the last architecture is devised using the Wald test design criterion where we do not exploit the maximum likelihood estimates
of the parameters under $H_1$, but those obtained by means of the recursive estimation algorithm.

As preliminary step towards the receiver derivations, let us define the following quantities. Specifically,
denote by $\bZ=[\bz_1 \ \bz_2]$ the primary data matrix and $\bZ_S=[\bz_{11} \ \ldots \ \bz_{1K} \ \bz_{21} \ \ldots \ \bz_{2K}]$ the 
overall matrix  of the training samples. Moreover, the probability density functions (pdfs) of $\bZ$ 
under $H_0$ and $H_1$ are given by
\be
f(\bZ;\bM,H_0)=\frac{\exp\left\{ -\frac{1}{2}\mbox{Tr}[\bM^{-1}\bZ\bZ^T] \right\}}{(2 \pi)^{N} \det(\bM)},
\ee
and
\begin{multline}
f(\bZ;\bM,\alpha_1,\alpha_2,H_1)
=\frac{1}{(2 \pi)^{N} \det(\bM)}
\\
\times \exp\Big\{ -\frac{1}{2}\mbox{Tr}\{\bM^{-1}[( \bz_1 - \alpha_1 \bv_1 + \alpha_2\bv_2)
(\bz_1 - \alpha_1 \bv_1 + \alpha_2\bv_2)^T
\\
+ (\bz_2 - \alpha_1 \bv_2 - \alpha_2\bv_1) (\bz_2 - \alpha_1 \bv_2 - \alpha_2\bv_1)^T ]\}\Big\},
\end{multline}
respectively, while the pdf of $\bZ_S$ is the same under both hypotheses, namely
\be
f(\bZ_S;\bM)=\frac{\exp\left\{ -\frac{1}{2}\mbox{Tr}[\bM^{-1}\bZ_S\bZ_S^T] \right\}}{(2 \pi)^{NK} \det^{K}(\bM)}.
\ee

\subsection{Two-Step GLRT}
Assume that $\bM$ is known, then the GLRT based upon primary data has the following form
\be
\frac{\dmax_{\alpha_1}\dmax_{\alpha_2} f(\bZ;\bM,\alpha_1,\alpha_2,H_1)}{f(\bZ;\bM,H_0)}\test \eta,
\label{eqn:2SGLRT_000}
\ee
where $\eta$ is the detection threshold\footnote{Hereafter, $\eta$ is used to denote the detection threshold or any proper modification of 
it for all the considered receivers.} chosen to ensure the desired level for the Probability of False Alarm ($P_{fa}$).
In Appendix \ref{appendix:derivation2SGLRT}, we show that (\ref{eqn:2SGLRT_000}) is equivalent to the following decision rule
\be
\frac{t_1(\bM) + t_2(\bM)}{\bv_1^T\bM^{-1}\bv_1 + \bv_2^T\bM^{-1}\bv_2} \test \eta,
\label{eqn:2SGLRT_00}
\ee
where $t_1(\bM)=(\bv_1^T\bM^{-1}\bz_1 + \bv_2^T\bM^{-1}\bz_2)^2$ and $t_2(\bM)=(\bv_1^T\bM^{-1}\bz_2 - \bv_2^T\bM^{-1}\bz_1)^2$.
The adaptivity is achieved replacing $\bM$ with
\be
\bS = \bZ_S\bZ_S^T,
\label{eqn:sampleCov}
\ee
namely $2K$-times the sample covariance matrix obtained from the secondary data.
Summarizing, the Two-Step GLRT is given by
\be
\tTSGLRT = \frac{t_1(\bS) + t_2(\bS)}{\bv_1^T\bS^{-1}\bv_1 + \bv_2^T\bS^{-1}\bv_2} \test \eta.
\label{eqn:SSAMF}
\ee
In the following, we refer to the above decision scheme as Symmetric Spectrum-AMF (SS-AMF).

\subsection{GLRT-based Receiver}
\label{Subsection:GLRTbased}
In this subsection, we propose an ad-hoc receiver exploiting the GLRT design idea. To this end, observe
that the GLR based upon primary and secondary data can be written as
\be
\frac{\dmax_{\alpha_1,\alpha_2}\dmax_{\bM} f(\bZ;\bM,\alpha_1,\alpha_2,H_1)f(\bZ_S;\bM)}{\dmax_{\bM} f(\bZ;\bM,H_0)f(\bZ_S;\bM)}
\ee
and involves the joint maximization of the pdfs with respect to $\alpha_i$, $i=1,2$, and $\bM$, which becomes
an intractable problem from a mathematical point of view. To redress this difficulty, we modify the GLRT approach
according to the following rationale:
\begin{enumerate}
\item assume that $\alpha_i$, $i=1,2$, are known and compute the GLRT, namely, perform the optimization with respect to $\bM$;
\item optimize the compressed likelihood function obtained at previous step with respect to $\alpha_i$, $i=1,2$, by means
of an iterative estimation algorithm.
\end{enumerate}
Thus, the $(K+1)$th root of the compressed likelihood functions under both hypotheses and with respect to $\bM$ are given by
\begin{align}
&f_1(\alpha_1,\alpha_2;\bZ)=[f(\bZ;\widehat{\bM}_1,\alpha_1,\alpha_2,H_1)f(\bZ_S;\widehat{\bM}_1)]^{1/(K+1)} \nonumber
\\
&\propto
\det[(\bz_1 - \bm_1(\alpha_1,\alpha_2))(\bz_1 - \bm_1(\alpha_1,\alpha_2))^T \nonumber
\\
&+
(\bz_2 - \bm_2(\alpha_1,\alpha_2))(\bz_2 - \bm_2(\alpha_1,\alpha_2))^T+\bS]^{-1}
\label{eqn:compressedH1}
\end{align}
under $H_1$, where
$ \bm_1(\alpha_1,\alpha_2)=\alpha_1 \bv_1 - \alpha_2\bv_2$, $\bm_2(\alpha_1,\alpha_2)=\alpha_1 \bv_2 + \alpha_2\bv_1$, and
\begin{multline}
f_0(\bZ)=[f(\bZ;\widehat{\bM}_0,H_0)f(\bZ_S;\widehat{\bM}_0)]^{1/(K+1)}
\\
\propto \frac{1}{\det[\bZ\bZ^T + \bS]},
\label{eqn:compressedH0}
\end{multline}
under $H_0$. In (\ref{eqn:compressedH1}) and (\ref{eqn:compressedH0}), $\bS$ 
is defined by (\ref{eqn:sampleCov}) and the estimates of $\bM$ under $H_0$ and $H_1$
are given by
\begin{multline}
\widehat{\bM}_1=\frac{1}{2K+2}[(\bz_1 - \bm_1(\alpha_1,\alpha_2))(\bz_1 - \bm_1(\alpha_1,\alpha_2))^T
\\
+
(\bz_2 - \bm_2(\alpha_1,\alpha_2))(\bz_2 - \bm_2(\alpha_1,\alpha_2))^T+\bS]
\label{eqn:MLestimateM}
\end{multline}
and $\widehat{\bM}_0=\frac{1}{2K+2}[\bZ\bZ^T + \bS]$,
respectively. It still remains to optimize $f_1(\alpha_1,\alpha_2;\bZ)$ over $\alpha_i$, $i=1,2$, which is tantamount
to solving
\begin{multline}
\dmin_{\alpha_1,\alpha_2} \det[(\bz_1 - \bm_1(\alpha_1,\alpha_2))(\bz_1 - \bm_1(\alpha_1,\alpha_2))^T
\\
+ (\bz_2 - \bm_2(\alpha_1,\alpha_2))(\bz_2 - \bm_2(\alpha_1,\alpha_2))^T+\bS].
\label{eqn:minimization00}
\end{multline}
Now, let us focus on the determinant of $(2K+2)\widehat{\bM}_1$ and observe that it can be written as
\begin{align}
&\det[\bS]\det[\bI + (\bZ-\bV)^T \bS^{-1}(\bZ-\bV)] \nonumber
\\
&=\det[\bS]\{[1+(\bz_1 - \bm_1(\alpha_1,\alpha_2))^T \bS^{-1} (\bz_1 - \bm_1(\alpha_1,\alpha_2))] \nonumber
\\
& \times [1+(\bz_2 - \bm_2(\alpha_1,\alpha_2))^T \bS^{-1} (\bz_2 - \bm_2(\alpha_1,\alpha_2))] \nonumber
\\
&-[(\bz_1 - \bm_1(\alpha_1,\alpha_2))^T \bS^{-1} (\bz_2 - \bm_2(\alpha_1,\alpha_2))]^2\} \nonumber 
\\
&=\det[\bS]h(\alpha_1,\alpha_2)
\end{align}
where $\bV=[\bm_1(\alpha_1,\alpha_2) \ \bm_2(\alpha_1,\alpha_2)]\in\R^{N\times 2}$. As a consequence, (\ref{eqn:minimization00})  amounts to
$\dmin_{\alpha_1,\alpha_2} h(\alpha_1,\alpha_2)$.
At this point, before describing the procedure aimed at finding the stationary points of $h(\alpha_1,\alpha_2)$,
an intermediate result is mandatory. More precisely, the following proposition holds true.
\begin{prop}
\label{proposition:coerciveFunction}
$h(\alpha_1,\alpha_2)$ is a coercive or radially unbounded function.
\end{prop}
\begin{proof}
See Appendix \ref{appendix:coerciveFunction}.
\end{proof}
As a consequence of the above proposition, 
the continuous and non-negative function $h(\alpha_1,\alpha_2)$ has a global minimum, which should be sought among its stationary points.
To find them, we follow an iterative procedure based on alternating (cyclic) optimizations, which estimates
$\alpha_m$, $m=1,2$, assuming that $\alpha_n$, $n\neq m$, $n=1,2$, is known. More precisely, let us start fixing $\alpha_2={\Halpha}_2^{(0)}$ with
${\Halpha}_2^{(0)}$ known\footnote{Notice that it is also possible to start from $\alpha_1={\Halpha}_1^{(0)}$ since the procedure is dual.}, 
and minimize $h_1(\alpha_1)=h(\alpha_1,{\Halpha}_2^{(0)})$ over $\alpha_1$.
To this end, we set the derivative of $h_1(\alpha_1)$ to zero and solve the following equation
\begin{align}
& \frac{d}{d\alpha_1}[h_1(\alpha_1)]=
-2\bv_1^\dag\bS^{-1}(\bz_1-\bm_1(\alpha_1,{\Halpha}_2^{(0)})) \nonumber
\\
&\times \left[1+(\bz_2-\bm_2(\alpha_1,
{\Halpha}_2^{(0)}))^\dag\bS^{-1}(\bz_2-\bm_2(\alpha_1,{\Halpha}_2^{(0)}))\right] \nonumber
\\
& -2\bv_2^\dag\bS^{-1}(\bz_2-\bm_2(\alpha_1,{\Halpha}_2^{(0)})) \nonumber
\\
&\times \left[1+(\bz_1-\bm_1(\alpha_1,{\Halpha}_2^{(0)}))^\dag\bS^{-1}
(\bz_1-\bm_1(\alpha_1,{\Halpha}_2^{(0)}))\right]\nonumber
\\
& -2(\bz_1-\bm_1(\alpha_1,{\Halpha}_2^{(0)}))^\dag\bS^{-1}(\bz_2-\bm_2(\alpha_1,{\Halpha}_2^{(0)})) \nonumber
\\
& \times \left[ -\bv^T_1\bS^{-1}(\bz_2-\bm_2(\alpha_1,{\Halpha}_2^{(0)})) \right. \nonumber
\\
&\left. - \bv^T_2\bS^{-1}(\bz_1-\bm_1(\alpha_1,{\Halpha}_2^{(0)})) \right]=0.
\end{align}
It is tedious but not difficult to show that the above equation can be recast as
\be
\frac{d}{d\alpha_1}[h_1(\alpha_1)]=\sum_{n=1}^4 b_i \alpha_1^{4-n}=0,
\label{eqn:thirdDegree}
\ee
where the expressions of the coefficients $b_i$, $i=1,\ldots,4$, are given in Appendix \ref{appendix:coefficients}.
Observe that (\ref{eqn:thirdDegree}) is a cubic equation with real coefficients and, hence, it admits at least one real solution.
The solutions of this equation can be explicitly obtained resorting to Cardano's method \cite{higherAlgebra} and we choose 
that real one, ${\Halpha}_1^{(1)}$, leading to the minimum of $h_1(\alpha_1)$.

Once ${\Halpha}_1^{(1)}$ is available, let $h_2(\alpha_2)=h(\Halpha_1^{(1)},\alpha_2)$ and repeat the same line of reasoning
used to find $\Halpha_1^{(1)}$ also for $\alpha_2$. In other words, we set to zero the derivative of $h_2(\alpha_2)$ with respect to $\alpha_2$,
namely
\begin{align}
&\frac{d}{d\alpha_2}[h_2(\alpha_2)]=
2\bv_2^\dag\bS^{-1}(\bz_1-\bm_1(\Halpha_1^{(1)},\alpha_2)) \nonumber
\\
& \times \left[1+(\bz_2-\bm_2(\Halpha_1^{(1)},
\alpha_2))^\dag\bS^{-1}(\bz_2-\bm_2(\Halpha_1^{(1)},\alpha_2))\right] \nonumber
\\
& -2\bv_1^\dag\bS^{-1}(\bz_2-\bm_2(\Halpha_1^{(1)},\alpha_2)) \nonumber
\\
&\times \left[1+(\bz_1-\bm_1(\Halpha_1^{(1)},\alpha_2))^\dag\bS^{-1}
(\bz_1-\bm_1(\Halpha_1^{(1)},\alpha_2))\right]\nonumber
\\
& -2(\bz_1-\bm_1(\Halpha_1^{(1)},\alpha_2))^\dag\bS^{-1}(\bz_2-\bm_2(\Halpha_1^{(1)},\alpha_2)) \nonumber
\\
& \times \left[ \bv^T_2\bS^{-1}(\bz_2-\bm_2(\Halpha_1^{(1)},\alpha_2)) \right. \nonumber
\\
&\left. - \bv^T_1\bS^{-1}(\bz_1-\bm_1(\Halpha_1^{(1)},\alpha_2)) \right]=0.
\end{align}
After algebraic manipulations of the above equation, we obtain another cubic equation
\be
\frac{d}{d\alpha_2}[h_2(\alpha_2)]=\sum_{n=1}^4 d_i \alpha_2^{4-n}=0.
\label{eqn:cubic2}
\ee
The real coefficients of the above equation are given in Appendix \ref{appendix:coefficients}.
Again, Cardano's method comes in handy to find closed-form expressions for the solutions of (\ref{eqn:cubic2}).
Among them, we choose the real one, ${\Halpha}_2^{(1)}$ say, leading to the minimum of $h_2(\alpha_2)$.

Generally speaking, the above iterations can be repeated to obtain values as fine 
as possible in the sense of the maximum likelihood estimation. Specifically, at the $n$th step,
the pair $({\Halpha}_1^{(n)},{\Halpha}_2^{(n)})$ is available.
Accordingly, we build the function $h_1(\alpha_1)$ (or $h_2(\alpha_2)$)
and find $\alpha_1$ (or $\alpha_2$) which returns the minimum value of $h_1(\alpha_1)$ (or $h_2(\alpha_2)$). This point represents the 
updated estimate of $\alpha_1$ (or $\alpha_2$) and can be used in the next iteration until a stopping condition is not satisfied.
The general procedure is described by Algorithm \ref{algorithm:AlgorithmPP}.

\begin{algorithm}
\caption{: Ad-hoc algorithm to estimate $(\alpha_1,\alpha_2)$.}
\begin{algorithmic}[1]
\REQUIRE $\bz_1$, $\bz_2$, $\bS$, $\bv_1$, $\bv_2$, $\epsilon_1$, $\epsilon_2$. 

\ENSURE ML estimates of $(\alpha_1,\alpha_2)$.

\STATE set $n=0$ , $l=1$ \OR $l=2$, 

\IF{$l=1$}
\STATE set $\Halpha_2^{(0)}=\bar{\alpha}_2$
\ELSE
\STATE set $\Halpha_1^{(0)}=\bar{\alpha}_1$
\ENDIF

\REPEAT 
\STATE{set $n=n+1$}
\IF{$l=1$}
\STATE{replace $\alpha_2$ in $h(\alpha_1,\alpha_2)$ with $\Halpha_2^{(n-1)}$} and compute $\Halpha_1^{(n)}$
\STATE use $\Halpha_1^{(n)}$ to update $\Halpha_2^{(n-1)}$ and to obtain $\Halpha_2^{(n)}$
\ELSE
\STATE{replace $\alpha_1$ in $h(\alpha_1,\alpha_2)$ with $\Halpha_1^{(n-1)}$} and compute $\Halpha_2^{(n)}$
\STATE use $\Halpha_2^{(n)}$ to update $\Halpha_1^{(n-1)}$ and to obtain $\Halpha_1^{(n)}$
\ENDIF
\UNTIL{$|\Halpha_1^{(n)}-\Halpha_1^{(n-1)}|>\epsilon_1$ \AND $|\Halpha_2^{(n)}-\Halpha_2^{(n-1)}|>\epsilon_2$}
\end{algorithmic}
\label{algorithm:AlgorithmPP}
\end{algorithm}

Let $\tilde{\alpha}_1$ and $\tilde{\alpha}_2$ the estimates of $\alpha_1$ and $\alpha_2$, respectively, at 
the $n$th iteration of the estimation algorithm, then
the final expression of the ad-hoc receiver is given by
\begin{multline}
t={\det[\bZ\bZ^T + \bS]}
\\
\times 
\det[(\bz_1 - \bm_1(\tilde{\alpha}_1,\tilde{\alpha}_2))(\bz_1 - \bm_1(\tilde{\alpha}_1,\tilde{\alpha}_2))^T
\\
+
(\bz_2 - \bm_2(\tilde{\alpha}_1,\tilde{\alpha}_2))(\bz_2 - \bm_2(\tilde{\alpha}_1,\tilde{\alpha}_2))^T+\bS]^{-1}
\test\eta.
\end{multline}
In the sequel, we will refer to this architecture as Iterative GLRT (I-GLRT).

Two remarks are now in order. First observe that it is possible to start the iterations exploiting the estimates of
$\alpha_1$ and $\alpha_2$ obtained by means of the two-step GLRT design procedure, which are given by (\ref{eqn:TSestimateA1}) 
and (\ref{eqn:TSestimateA2}). Second, the iterative procedure yields a sequence of estimates
\be
(\Halpha_1^{(1)},\Halpha_2^{(1)}), \ (\Halpha_1^{(2)},\Halpha_2^{(2)}),
\ (\Halpha_1^{(3)},\Halpha_2^{(3)}), \dots, (\Halpha_1^{(n)},\Halpha_2^{(n)}),\ldots
\label{eqn:sequence1}
\ee
which shares an important property shown in the following.

\begin{prop}
\label{proposition:convergence}
From the sequence (\ref{eqn:sequence1}) it is possible to extract a 
subsequence that converges to a stationary point of $h(\alpha_1,\alpha_2)$.
\end{prop}
\begin{proof}
See Appendix \ref{appendix:convergence}.
\end{proof}
The proof of this proposition highlights that (\ref{eqn:sequence1})
induces a decreasing sequence of objective function $h(\alpha_1,\alpha_2)$ values. This implies that
if we use $(\Halpha_1^{\mbox{\tiny TS}}(\bS),\Halpha_2^{\mbox{\tiny TS}}(\bS))$, given by (\ref{eqn:TSestimateA1}) and (\ref{eqn:TSestimateA2}), respectively, 
as starting point of the algorithm, then it is possible to attain better estimates of $(\alpha_1,\alpha_2)$ in the 
sense of the likelihood optimization. Interestingly, they also lead to better
detection performances as it will be shown in Section \ref{Sec:PerformanceAnalysis}.
Additionally, it is worth noticing that when $(\Halpha_1^{\mbox{\tiny TS}}(\bS),\Halpha_2^{\mbox{\tiny TS}}(\bS))$ is contained
within a suitable neighborhood of the global minimum, it may happen that the optimum value belongs to the trajectory described by
(\ref{eqn:sequence1}) and the I-GLRT becomes asymptotically equivalent to the plain GLRT.
On the other hand, the asymptotic estimates provided by Algorithm 1 are the coordinates of a stationary point 
that could be either a local minimum or a saddle point.

\subsection{Rao test}
In Section \ref{Sec:ProblemFormulation}, we have observed that 
the relevant parameter to the decision problem (\ref{eqn:hypothesistest01}) is given by the vector
$\btheta_A=[\alpha_1 \ \alpha_2]^T$, while the elements of $\bM$ represent the nuisance parameters.
Moreover, since $\bM\in \bS_{++}^N$, it can be well-represented by the $[(N-1)N/2]$-dimensional vector
$\btheta_B=\bff(\bM)\in\R^{(N+1)N/2}$,
where $\bff(\cdot)$ is a vector-valued function that selects in unequivocal way (bijection) the elements of a symmetric matrix.
Let $\btheta=[\btheta_A^T \ \btheta_B^T]^T$ the overall parameter vector for the problem at hand and denote
by $\btheta_0$ the estimate of $\btheta$ under the $H_0$ hypothesis. It is evident that
$\btheta_0=\left[ \bzero \ \bff(\bS_0) \right]^T$,
where $\bS_0=\bZ\bZ^T+\bS$. Finally, let us partition the Fisher information matrix as follows
% \begin{align}
% \bF(\btheta)&=-E\left[
% \frac{\partial^2}{\partial \btheta \partial \btheta^T} \log\big(f(\bZ;\bM,\alpha_1,\alpha_2,H_1)f(\bZ_S;\bM)\big)
% \right]\nonumber
% \\
% &=\left[\begin{array}{cc}
% \bF_{AA}(\btheta) & \bF_{AB}(\btheta)
% \\
% \bF_{BA}(\btheta) & \bF_{BB}(\btheta)
% \end{array}\right],
% \end{align}
\be
\bF(\btheta)=\left[\begin{array}{cc}
\bF_{AA}(\btheta) & \bF_{AB}(\btheta)
\\
\bF_{BA}(\btheta) & \bF_{BB}(\btheta)
\end{array}\right],
\ee
where
\begin{multline}
\bF_{XY}(\btheta)
\\
=-E\left[
\frac{\partial^2}{\partial \btheta_X \partial \btheta_Y^T} \log \big(f(\bZ;\bM,\alpha_1,\alpha_2,H_1)f(\bZ_S;\bM)\big)
\right]
\\
\in\R^{x\times y}.
\label{eqn:partitionFisher}
\end{multline}
In (\ref{eqn:partitionFisher}), $(X,Y)\in\{A,B\}\times \{A,B\}$, 
\be
x=\left\{
\begin{array}{ll}
2, & \mbox{if } X=A,
\\
(N-1)N/2, & \mbox{if } X=B,
\end{array}
\right.
\ee
and
\be
y=\left\{
\begin{array}{ll}
2, & \mbox{if } Y=A,
\\
(N-1)N/2, & \mbox{if } Y=B.
\end{array}
\right.
\ee
With the above definitions in mind, we can provide the expression of the Rao test for the problem at hand \cite{Kay-SSP-DT}
\begin{multline}
\left\{\frac{\partial}{\partial \btheta_A}\big[\log\big(f(\bZ;\bM,\alpha_1,\alpha_2,H_1)f(\bZ_S;\bM)\big)\big]\right\}^T_{\btheta=\btheta_0}
\\
\times\left\{[\bF(\btheta)^{-1}]_{AA}\right\}_{\btheta=\btheta_0}
\\
\times\left\{\frac{\partial}{\partial \btheta_A}\big[\log\big(f(\bZ;\bM,\alpha_1,\alpha_2,H_1)f(\bZ_S;\bM)\big)\big]\right\}_{\btheta=\btheta_0},
\end{multline}
where $[\bF(\btheta)^{-1}]_{AA}$ is the sub-block of the inverse of the Fisher information matrix formed by selecting
its first two rows and the first two columns; in addition it can be written as
\be
[\bF(\btheta)^{-1}]_{AA} = [\bF_{AA}(\btheta) - \bF_{AB}(\btheta)\bF^{-1}_{BB}(\btheta)\bF_{BA}(\btheta)]^{-1}.
\label{eqn:FAAmeno1}
\ee
In Appendix \ref{appendix:derivationRao}, it is proved that the Rao test can be simplified as follows
% \be
% \tRAO = \frac{(\bv_1^T\bS_0^{-1}\bz_1 + \bv_2^T\bS_0^{-1}\bz_2)^2
% +(\bv_1^T\bS_0^{-1}\bz_2-\bv_2^T\bS_0^{-1}\bz_1)^2}
% {\bv_1^T\bS_0^{-1}\bv_1 + \bv_2^T\bS_0^{-1}\bv_2}\test\eta.
% \label{eqn:raoTest}
% \ee
\be
\tRAO = \frac{t_1(\bS_0)+t_2(\bS_0)}
{\bv_1^T\bS_0^{-1}\bv_1 + \bv_2^T\bS_0^{-1}\bv_2}\test\eta.
\label{eqn:raoTest}
\ee
In the following, we refer to the above detector as Symmetric Spectrum-RAO test (SS-RAO). 
Observe that the decision statistic in (\ref{eqn:raoTest}) is none other than 
that in (\ref{eqn:SSAMF}) with $\bS_0$, which is the sample covariance matrix over both primary and secondary data, in place 
of $\bS$, namely the sample covariance matrix over the secondary data. This similarity is also encountered 
in the non-symmetric case where the Rao test \cite{DeMaio-RAO} and the AMF \cite{Kelly-Nitzberg} share the same expression
but for the sample covariance matrix.

\subsection{Receiver based upon the Wald test}
In the presence of nuisance parameters, the Wald test exhibits the following form \cite{Kay-SSP-DT}
\be
\left(\btheta_{A,1}-\btheta_{A,H_0}\right)^T
\left\{[\bF(\btheta)^{-1}]_{AA}\right\}_{\btheta=\btheta_1}^{-1}
\left(\btheta_{A,1}-\btheta_{A,H_0}\right)\test \eta,
\ee
where $\btheta_{A,H_0}=[0 \ 0]^T$ is the value of the relevant parameters under $H_0$, 
$\btheta_{A,1}=[\tilde{\alpha}_1 \ \tilde{\alpha}_2]^T\in \R^{2\times 1}$ is the maximum likelihood
estimate of the relevant parameter under $H_1$, and $\btheta_1$ is the maximum likelihood estimate of the entire parameter vector
under $H_1$, namely $\btheta_1=\left[\btheta_{A,1} \ \bff(\widehat{\bM}_1(\tilde{\alpha}_1,\tilde{\alpha}_2)) \right]^T$
% \be
% \btheta_1=\left[
% \begin{array}{c}
% \btheta_{A,1}
% \\
% \bff(\widehat{\bM}_1(\tilde{\alpha}_1,\tilde{\alpha}_2))
% \end{array}
% \right]
% \ee
with $\widehat{\bM}_1(\tilde{\alpha}_1,\tilde{\alpha}_2)$ the maximum likelihood estimate of the interference covariance matrix under $H_1$.
Since closed form expression for $\btheta_1$ is not available, we resort to $\tilde{\alpha}_1$ and $\tilde{\alpha}_2$ given in 
Section \ref{Subsection:GLRTbased} in place of the exact maximum likelihood estimates.
Thus, the approximated Wald test, also referred to in the following as Iterative Wald test (I-WALD), becomes
\be
\left[ \Halpha_1^{(N)} \ \Halpha_2^{(N)} \right]
\left\{[\bF(\btheta)^{-1}]_{AA}\right\}_{\btheta=\bar{\btheta}_1}^{-1}
\left[ 
\begin{array}{c}
\Halpha_1^{(N)} 
\\ 
\Halpha_2^{(N)}
\end{array}
\right]\test \eta,
\ee
where $(\Halpha_1^{(N)},\Halpha_2^{(N)})$ are obtained after $2N$ iterations of Algorithm \ref{algorithm:AlgorithmPP} using as
initial seed the $\Halpha_1^{\mbox{\tiny TS}}(\bS)$  or $\Halpha_2^{\mbox{\tiny TS}}(\bS)$,
$\bar{\btheta}_1=\left[ \Halpha_1^{(N)} \ \Halpha_2^{(N)} \ \bff(\widehat{\bM}_1(\Halpha_1^{(N)},\Halpha_2^{(N)})) \right]^T$,
% \be
% \bar{\btheta}_1=\left[
% \begin{array}{c}
% \Halpha_1^{(N)}
% \\
% \Halpha_2^{(N)}
% \\
% \bff(\widehat{\bM}_1(\Halpha_1^{(N)},\Halpha_2^{(N)}))
% \end{array}
% \right],
% \ee
and $\widehat{\bM}_1(\Halpha_1^{(N)},\Halpha_2^{(N)})$ is given by (\ref{eqn:MLestimateM}) with $(\Halpha_1^{(N)},\Halpha_2^{(N)})$
in place of $(\alpha_1,\alpha_2)$.

Finally, exploiting results contained in Appendix \ref{appendix:derivationRao}, it straightforward
to obtain that
$\left\{[\bF(\btheta)^{-1}]_{AA}\right\}_{\btheta=\bar{\btheta}_1}^{-1}=\sigma_F \bI_2$,
% \left[
% \begin{array}{cc}
% X & \bzero
% \\
% \bzero & X
% \end{array}
% \right],
where $\sigma_F = \bv_1^T\left[\widehat{\bM}_1(\Halpha_1^{(N)},\Halpha_2^{(N)})\right]^{-1}\bv_1 + \bv_2^T\left[\widehat{\bM}_1(\Halpha_1^{(N)},\Halpha_2^{(N)})\right]^{-1}\bv_2$.

\section{Illustrative Examples}
\label{Sec:PerformanceAnalysis}
In this section, we investigate the detection performances of the previously devised detectors in comparison with conventional architectures
that do not exploit the symmetric spectral properties of the interference. The considered competitors are Kelly's GLRT (K-GLRT) \cite{Kelly86},
the Adaptive Matched Filter\footnote{Note that the AMF coincides with the Wald test for problem (\ref{eqn:hypothesistest00}) \cite{DeMaio-WALD-AMF}.} 
(AMF), and the Rao test \cite{DeMaio-RAO}. The analysis is conducted resorting to both simulated and live recorded data. 

\subsection{Simulated data}
Since closed form expressions for the Probability of Detection ($P_d$) and the Probability of False Alarm ($P_{fa}$) are not available,
the numerical examples are obtained by means of standard Monte Carlo counting techniques. Specifically,
we compute the thresholds necessary to ensure a preassigned value of $P_{fa}$ and $P_d$ resorting 
to $100/P_{fa}$ and $10^4$ independent trials, respectively.

The interference is modeled as a circular complex normal random vector with the following covariance matrix
$\bM = \sigma^2_n \bI_N + \sigma^2_c \bM_c$,
where $\sigma^2_n=1$, $\sigma^2_c>0$ is evaluated assuming a clutter-to-noise ratio of 20 dB,
the $(i,j)$th element of $\bM_c$ is given by $\rho_c^{|i-j|^2}e^{j2\pi f_d (i-j)}$
% $\rho_c^{|i-j|}$ 
with $\rho_c=0.9$ and $f_d$ the Doppler frequency of the clutter.
Moreover, we assume that the system exploits $N$ temporal channels and that the Signal-to-Noise-plus-Interference Ratio (SINR) shares the following
expression $\mbox{SINR} = |\alpha|^2 \bv(\nu_d)^\dag \bM^{-1} \bv(\nu_d)$,
where the temporal steering vector $\bv(\nu_d)$ is given by
$\bv(\nu_d) = \frac{1}{\sqrt{N}} \left[ 1 \ e^{j2\pi \nu_d} \ \ldots \ e^{j2\pi (N-1) \nu_d} \right]^T$
with $\nu_d$ the normalized Doppler frequency\footnote{Observe that we do not distinguish between the actual and the nominal steering vector
because we assume perfectly matching conditions.}. Finally, in all the considered examples, we set $P_{fa}=10^{-4}$ and 
plot, for comparison purposes, the GLRT for known interference covariance matrix also referred to as benchmark detector.

Before proceeding with the performance comparisons, we establish an empirical criterion to set the number of iterations, $n$ say, required
by Algorithm 1 to provide reliable estimates and by I-GLRT and I-WALD to achieve the best performances. To this end, in Figures \ref{fig:figure01} and 
\ref{fig:figure02} we plot $P_d$ versus SINR for the I-GLRT and the I-WALD assuming
different values of $n$; both figures consider $N=8$, $K=6$, and $\nu_d=0.1$. 
Their inspection highlights that the upper bound on the performance is achieved
exploiting $3\leq n \leq 5$ iterations. 
Moreover, values of $n$ greater than 1 confer to the estimation procedure a more robust behavior with respect to 
possible outliers as shown in 
Figures \ref{fig:figure03} and \ref{fig:figure03a}. Therein, we compare the estimates of $\alpha_1$ and $\alpha_2$ provided 
by Algorithm 1 ($n=5$) with those obtained by means of the two-step design procedure, 
namely $(\Halpha_1^{\mbox{\tiny TS}}(\bS),\Halpha_2^{\mbox{\tiny TS}}(\bS))$. 
The actual values of $\alpha_1$ and $\alpha_2$ are plotted too.
The estimates are obtained using as initial seed of 
Algorithm 1 the estimate $\Halpha_1^{\mbox{\tiny TS}}(\bS)$  or $\Halpha_2^{\mbox{\tiny TS}}(\bS)$.
Finally, further simulation results, not reported here for the sake of brevity, shows that when $\nu_d=0$
the curves of $P_d$ referring to different values of $n$ are one and the same.
In the remaining numerical examples, we use $n=3$, which represents a reasonable trade off between robustness, 
detection performance, and computational load in
different operating conditions.

In Figure \ref{fig:figure04}, we show the performance of the I-GLRT, the I-WALD, the SS-AMF, and the SS-RAO in sample-starved scenarios, namely when
the number of secondary data is lower than the vector size ($K<N$). To this end, we 
set $N=8$ and $K=6$. The plots highlight that the best detection performances are ensured by the I-GLRT with a gain of about
$5$ dB over the SS-AMF and the I-WALD, which share the same behavior, whereas, for the considered values of the parameters,
the SS-RAO is useless since its $P_d$ does not achieve values greater than 0.1 (this was expected due to the small number of training data).

Figures \ref{fig:figure05}, \ref{fig:figure06}, and \ref{fig:figure07} refer to the cases $N<K\leq 2N$, $K=2N$, and $K=4N$, respectively.
Again, the I-GLRT outperforms the other decision structures.
Moreover, observe that the I-GLRT, the SS-AMF, and the I-WALD can ensure performance gains within $1$ dB (as shown in Figure \ref{fig:figure07})
and $5$ dB (as shown in Figure \ref{fig:figure05}) with respect to their natural competitors.
Another important remark concerns the performance hierarchy that keeps unaltered as $K$ increases.
Finally, observe that, for the considered parameter setting, the SS-RAO and its conventional competitor 
achieve $P_d=1$ at reasonable SINR values only when $N=4K$.

\subsection{Real data}
The aim of this section is two-fold. First, we study the CFAR behavior of the introduced detectors
in the presence of live symmetric clutter data which might also deviate from the receivers design hypotheses, then we 
assess their detection performance.
To this end, we exploit the MIT-LL Phase-One radar dataset, which contains land clutter and refers to different bands, polarizations,
range resolutions, and scanning modes.
Each data file is composed of
$N_t$ temporal returns from $N_s$ range cells which are
stored in an  $N_t \times N_s$ complex matrix. Further details on the description of the dataset can be found
in \cite[and references therein]{DeMaio-Foglia}.

Let us begin with the CFAR analysis and set the threshold of the receivers to return $P_{fa}=10^{-4}$ assuming
spatially homogeneous white Gaussian clutter.
These thresholds are exploited to evaluate the actual $P_{fa}$ when the detectors operate in the
presence of measured clutter data. The procedure we adopt to
select the primary and secondary data employed
for computing a realization of the decision statistics is pictorially
described in Figure \ref{fig:figure08}, where the primary cell is denoted by $P_c$ and the set of secondary
data is composed of  $K$ cells ($K$ even) whose number ranges from $P_c-\frac{K}{2}$ to $P_c-1$
and between $P_c+1$ and $P_c+\frac{K}{2}$. In other words the training set
contains the returns from the $\frac{K}{2}$ cells on the left of
$P_c$ and the returns from the $\frac{K}{2}$ cells on the right of the
cell under test.
The $N \times (K+1)$ data window is slided in both time and space until the end of the dataset.
By doing so, the total number of different
data windows is
$N_{trials}=(N_t-N+1)\times(N_s-K)$
coinciding with the
number of trials available to estimate the actual $P_{fa}$.
Otherwise stated, for each data window, we perform the four statistical tests
and, for each of them, we count the number of false alarms $N_{fa}$. The actual false alarm probability,
$\widehat{P}_{fa}$ say, is thus evaluated as
\[
\widehat{P}_{fa}=\displaystyle{\frac{N_{fa}}{N_{trials}}}=
\displaystyle{\frac{\displaystyle{\sum_{k=\frac{K}{2}+1}^{N_s-\frac{K}{2}}} N_{fa}(k)}{N_{trials}}}
\]
where $N_{fa}(k)$ denotes the number of false alarms resulting when $P_c=k$, namely
when the cell under test is the $k$-th range bin.

As to the temporal steering vector, $\nu_d$ is chosen equal to $0$ Hz
in order to simulate a very challenging condition
of a possible target in deep clutter. The results are reported in Table~\ref{landCFAR} for
both HH and VV polarimetric channels.
\begin{table}[ht!]
\begin{center}
\begin{tabular}{||c|c||c|c|c||}
\hline
\hline    Receiver   & Polarization  & $K=6$ & $K=12$ & $K=16$   \\
\hline \hline
     Kelly's GLRT    & HH            & N.A.  &  1.1   & 2.3 \\
                     & VV            & N.A.  &  1.5   & 3.2 \\
\hline\hline
     AMF             & HH            & N.A.  &  1.3   & 1.9 \\
                     & VV            & N.A.  &  1.8   & 3.9 \\
\hline\hline
     RAO    	     & HH            & N.A.  &  1.8   & 2.3 \\
                     & VV            & N.A.  &  2.1   & 2.7 \\
\hline \hline
     I-GLRT          & HH            &  0.9   & 1.46  & 2.7   \\
    ($n=3$)    	     & VV            &  0.9   & 2.3   & 4.3   \\
\hline \hline
     SS-AMF          & HH            &  0.9   & 1.7   & 2.9      \\
		     & VV            &  1.0   & 3.3   & 6.0    \\
\hline\hline
     I-WALD          & HH            &  0.9   & 1.8   & 2.9      \\
     ($n=3$)         & VV            &  1.0   & 3.3   & 6.0    \\
\hline\hline
     SS-RAO          & HH            &  1.1   & 2.1   & 3.4      \\
		     & VV            &  1.2   & 3.4   & 5.0   \\
\hline\hline
\end{tabular}
\end{center}
\caption{$\widehat{P}_{fa}/10^{-4}$ for $N=8$ and three values of $K$.}
\label{landCFAR}
\end{table}
They show that for $K=6$ the I-GLRT, the SS-AMF, the I-WALD, and the SS-RAO nominally behave
in terms of $P_{fa}$, while as $K$ increases all the receivers exhibit a slight mismatch between $\widehat{P}_{fa}$
and the nominal $P_{fa}$. Notice also that the VV channel $\widehat{P}_{fa}$ is higher
than the HH one for all the considered experiments on real data. 
This could be practically justified observing that the reflectivity on the 
HH polarization is usually a few dB lower than that on the VV channel (see Chapter 5 of \cite{Richards}).

Finally, the $P_d$ curves versus the SINR obtained using the live dataset are shown 
in Figures \ref{fig:figure09} and \ref{fig:figure10}. They agree with the hierarchy observed on simulated data with
the I-GLRT providing the best performance.

\section{Conclusions}
We have devised four different decision schemes which take advantage of some spectral properties of the 
clutter usually arising in a ground clutter environment.
Specifically, at the design stage we have assumed an interference PSD real and even with the consequence that
the resulting covariance matrix is real. This seemingly minor feature reduces the number of unknowns and 
allows to recast the problem at hand in terms of 
statistically independent and real quantities that can be suitably exploited for estimation purposes. As a matter of fact,
the architectures devised under the above assumptions are capable of guaranteeing reasonable detection performances 
also in sample-starved scenarios. The applied design criteria are the Rao test, the two-step GLRT, and suboptimum
modifications of both the plain GLRT and the Wald test.
Remarkably, these suboptimum procedures rely on an alternating estimation algorithm for the target response that ensures quality guarantee.
Additionally, we have shown that the estimates provided by the above algorithm are asymptotically equivalent to the maximum likelihood estimates.
In order to prove the effectiveness of this approach, the numerical examples shown in Section IV make use of both 
simulated data and live recorded data. More precisely, we have resorted to the MIT-LL Phase-One radar dataset, 
which contains land clutter. The analysis has highlighted the superiority of the newly proposed architectures over
the conventional detectors which do not capitalize on the real and even PSD of the clutter. It is also important to remark
that the performances on live recorded data are in agreement with those obtained on simulated data.

Future research tracks might concern the extension of the proposed framework to the case of heterogeneous ground clutter, where
the interference in primary and secondary data share the same covariance structure but different power levels. Finally,
it could be of interest conceiving an automatic spectrum analyzer that is capable to establish whether or not the clutter spectrum
shares symmetry properties. Then, according to the clutter properties, this decision scheme triggers either a conventional receiver 
or a newly proposed architecture.

% \vfill
% \pagebreak

\appendices

\section{Derivation of (\ref{eqn:2SGLRT_00})}
\label{appendix:derivation2SGLRT}

In order to find the stationary points of
\be
g(\alpha_1,\alpha_2)=\frac{f(\bZ;\bM,\alpha_1,\alpha_2,H_1)}{f(\bZ;\bM,H_0)},
\ee
we observe that, since the natural logarithm, $\ln(\cdot)$ say, is
an increasing function of the argument, the following equality holds true:
$\arg\max_{\alpha_1,\alpha_2} g(\alpha_1,\alpha_2) = \arg\max_{\alpha_1,\alpha_2} \ln [g(\alpha_1,\alpha_2)]$,
where
\begin{multline}
\ln [g(\alpha_1,\alpha_2)]
=-\frac{1}{2}\{
( \bz_1 - \alpha_1 \bv_1 + \alpha_2\bv_2)^T\bM^{-1}
\\
\times( \bz_1 - \alpha_1 \bv_1 + \alpha_2\bv_2) 
+(\bz_2 - \alpha_1 \bv_2 - \alpha_2\bv_1)^T\bM^{-1}
\\
\times(\bz_2 - \alpha_1 \bv_2 - \alpha_2\bv_1)
\left.-\bz_1^T\bM^{-1}\bz_1-\bz_2^T\bM^{-1}\bz_2
\right\}.
\label{eqn:naturalLog2SGLRT}
\end{multline}
In addition, the test
\be
g(\alpha_1,\alpha_2)\test \eta
\ee
is statistically equivalent to
\be
\ln[g(\alpha_1,\alpha_2)]\test \eta.
\ee
Now, set to zero the gradient of $\ln [g(\alpha_1,\alpha_2)]$ to obtain the following system of equations
\be
\left\{
\begin{array}{l}
\ds\frac{\partial \ln [g(\alpha_1,\alpha_2)]}{\partial \alpha_1} = 0,
\\
\ds \frac{\partial \ln [g(\alpha_1,\alpha_2)]}{\partial \alpha_2} = 0.
\end{array}
\right.
\ee
It is not difficult to show that the solutions of the above system are
\begin{align}
\Halpha_1^{\mbox{\tiny TS}}(\bM) &= \frac{ \bv_1^T\bM^{-1}\bz_1+\bv_2^T\bM^{-1}\bz_2 }{\bv_1^T\bM^{-1}\bv_1+\bv_2^T\bM^{-1}\bv_2},
\label{eqn:TSestimateA1}
\\
\Halpha_2^{\mbox{\tiny TS}}(\bM) &= \frac{ \bv_1^T\bM^{-1}\bz_2+\bv_2^T\bM^{-1}\bz_1 }{\bv_1^T\bM^{-1}\bv_1+\bv_2^T\bM^{-1}\bv_2},
\label{eqn:TSestimateA2}
\end{align}
which replaced in (\ref{eqn:naturalLog2SGLRT}) yields
\be
\frac{1}{2}
\frac{(\bv_1^T\bM^{-1}\bz_1 + \bv_2^T\bM^{-1}\bz_2)^2 +
(\bv_1^T\bM^{-1}\bz_2 - \bv_2^T\bM^{-1}\bz_1)^2}{\bv_1^T\bM^{-1}\bv_1 + \bv_2^T\bM^{-1}\bv_2}.
\ee

\section{Proof of Proposition \ref{proposition:coerciveFunction}}
\label{appendix:coerciveFunction}
$h(\alpha_1,\alpha_2)$ can be written in terms of 
$\bz_{1w}=\bS^{-1/2}\bz_1$, $\bz_{2w}=\bS^{-1/2}\bz_2$, $\bv_{1w}=\bS^{-1/2}\bv_1$, and $\bv_{2w}=\bS^{-1/2}\bv_2$
as follows
\begin{multline}
h(\alpha_1,\alpha_2)=[1+(\bz_{1w} - \bm_{1w}(\alpha_1,\alpha_2))^T (\bz_{1w} - \bm_{1w}(\alpha_1,\alpha_2))]
\\
\times [1+(\bz_{2w} - \bm_{2w}(\alpha_1,\alpha_2))^T  (\bz_{2w} - \bm_{2w}(\alpha_1,\alpha_2))] \nonumber
\\
-[(\bz_{1w} - \bm_{1w}(\alpha_1,\alpha_2))^T  (\bz_{2w} - \bm_{2w}(\alpha_1,\alpha_2))]^2,
\end{multline}
where $\bm_{1w}(\alpha_1,\alpha_2) =\alpha_1 \bv_{1w} - \alpha_2\bv_{2w}$ and 
$\bm_{2w}(\alpha_1,\alpha_2) =\alpha_1 \bv_{2w} + \alpha_2\bv_{1w}$.
Observe that $h(\alpha_1,\alpha_2)$ is continuous and is given by the	 ratio between the determinant of two positive definite matrices. As a result, it is strictly 
positive, namely
$h(\alpha_1,\alpha_2)>0$, $\forall (\alpha_1,\alpha_2)\in \R^2$. 

Now, exploit the Schwartz inequality to show that $\forall (\alpha_1,\alpha_2)\in\R^2$
\begin{align}
& h(\alpha_1,\alpha_2) \geq [1+\| \bz_{1w} - \bm_{1w}(\alpha_1,\alpha_2) \|^2] \nonumber
\\
& \times [1+\| \bz_{2w} - \bm_{2w}(\alpha_1,\alpha_2) \|^2] \nonumber
\\
&- \| \bz_{1w} - \bm_{1w}(\alpha_1,\alpha_2) \|^2 \|\bz_{2w} - \bm_{2w}(\alpha_1,\alpha_2)\|^2 
\\
&=1+ \| \bz_{1w} - \bm_{1w}(\alpha_1,\alpha_2) \|^2 + \| \bz_{2w} - \bm_{2w}(\alpha_1,\alpha_2) \|^2 
\\
& > \| \bz_{1w} - \bm_{1w}(\alpha_1,\alpha_2) \|^2 + \| \bz_{2w} - \bm_{2w}(\alpha_1,\alpha_2) \|^2
\\
&= \|\bz_{1w}\|^2 + \|\bm_{1w}(\alpha_1,\alpha_2) \|^2 - 2 \bz_{1w}^T\bm_{1w}(\alpha_1,\alpha_2) \nonumber 
\\
&+ \|\bz_{2w}\|^2 + \|\bm_{2w}(\alpha_1,\alpha_2) \|^2 - 2 \bz_{2w}^T\bm_{2w}(\alpha_1,\alpha_2)
\label{eqn:corciveIntermed00}
\\
& > \|\bm_{1w}(\alpha_1,\alpha_2) \|^2 - 2 \bz_{1w}^T\bm_{1w}(\alpha_1,\alpha_2) \nonumber
\\
&+\|\bm_{2w}(\alpha_1,\alpha_2) \|^2 - 2 \bz_{2w}^T\bm_{2w}(\alpha_1,\alpha_2)
\label{eqn:corciveIntermed01}
\\
&=(\alpha_1^2+\alpha_2^2)(\|\bv_{1w}\|^2+\|\bv_{2w}\|^2) \nonumber
\\
& -2\alpha_1(\bz_{1w}^T\bv_{1w}+\bz_{2w}^T\bv_{2w})
+2\alpha_2(\bz_{1w}^T\bv_{2w}-\bz_{2w}^T\bv_{1w}),
\end{align}
where the inequality between (\ref{eqn:corciveIntermed00}) and (\ref{eqn:corciveIntermed01}) is due to the fact that 
$\|\bz_{iw}\|^2>0$, $i=1,2$.

% Now, to further simplify the expression of $h((\alpha_1,\alpha_2))$, define the following
% quantities
% \begin{align}
% x(\alpha_1) &= (\bz_{1w}-\alpha_1\bv_{1w})^T (\bz_{1w}-\alpha_1\bv_{1w})=\|\bz_{1w}-\alpha_1\bv_{1w}\|^2 > 0,
% \\
% y(\alpha_2) &= (\bz_{2w}-\alpha_2\bv_{2w})^T (\bz_{2w}-\alpha_2\bv_{2w})=\| \bz_{2w}-\alpha_2\bv_{2w} \|^2 > 0,
% \end{align}
% and exploit the Schwartz inequality to show that
% \begin{align}
% h(\alpha_1,\alpha_2) &\geq [1+\|\bz_{1w}-\alpha_1\bv_{1w}\|^2] 
% [1+\|\bz_{2w}-\alpha_2\bv_{2w}\|^2] - \|\bz_{1w}-\alpha_1\bv_{1w}\|^2\|\bz_{2w}-\alpha_2\bv_{2w}\|^2 \nonumber
% \\
% &=[1+x(\alpha_1)][1+y(\alpha_2)]-x(\alpha_1)y(\alpha_2)=1+x(\alpha_1)+y(\alpha_2) > x(\alpha_1)+y(\alpha_2).
% \end{align}
As next step, let us define
$x(\alpha_1) = \alpha_1^2(\|\bv_{1w}\|^2+\|\bv_{2w}\|^2)-2\alpha_1(\bz_{1w}^T\bv_{1w}+\bz_{2w}^T\bv_{2w})$ and
$y(\alpha_2) = \alpha_2^2(\|\bv_{1w}\|^2+\|\bv_{2w}\|^2)+2\alpha_2(\bz_{1w}^T\bv_{2w}-\bz_{2w}^T\bv_{1w})$,
and observe that $h(\alpha_1,\alpha_2)>x(\alpha_1)+y(\alpha_2)$.
It is clear that
$\lim_{\alpha_1\rightarrow \pm\infty} x(\alpha_1) = +\infty$ and $\lim_{\alpha_2\rightarrow \pm\infty} y(\alpha_2) = +\infty$.
Gathering the above results yields
$\lim_{\|\balpha\|\rightarrow +\infty} h(\alpha_1,\alpha_2)=+\infty$,
% \label{eqn:limiteCoercive}
where $\balpha=[\alpha_1 \ \alpha_2]^T$. Thus, by definition, $h(\alpha_1,\alpha_2)$ is a coercive or radially unbounded function.

\section{Expressions of the Coefficients for equations (\ref{eqn:thirdDegree}) and (\ref{eqn:cubic2})}
\label{appendix:coefficients}
The coefficients of equation (\ref{eqn:thirdDegree}) are give by
\begin{align}
b_1 &= a_1 a_3 + a_9 a_6 - 2 a_{11} a_{14}, \nonumber
\\
b_2 &= a_1 a_4 + a_2 a_3 + a_6 a_{10} + a_9 a_7 - 2 a_{11} a_{15} - 2 a_{12} a_{14}, \nonumber
\\
b_3 &= a_1 a_5 + a_2 a_4 + a_7 a_{10} + a_9 a_8 - 2 a_{12} a_{15} - 2 a_{13} a_{14}, \nonumber
\\
b_4 &= a_2 a_5 + a_8 a_{10} - 2 a_{13} a_{15},
\end{align}
where
\begin{align}
& a_1 = 2 \bv_1^T \bS^{-1} \bv_1, \ a_2 = -2 {\Halpha}_2^{(0)} \bv_1^T \bS^{-1} \bv_2 - 2 \bv_1^T \bS^{-1} \bz_1, \nonumber
\\
& a_3 = \bv_2^T \bS^{-1} \bv_2, \ a_4 = -2  \bv_2^T \bS^{-1} \bz_2  + 2 {\Halpha}_2^{(0)}  \bv_2^T \bS^{-1} \bv_1, \nonumber
\\
& a_5 = ({\Halpha}_2^{(0)})^2 \bv_1^T \bS^{-1} \bv_1 - 2 {\Halpha}_2^{(0)}  \bv_1^T \bS^{-1} \bz_2  +   \bz_2^T \bS^{-1} \bz_2  + 1, \nonumber
\\
& a_6 = a_1/2, \ a_7 = -2  \bv_1^T \bS^{-1} \bz_1  - 2 {\Halpha}_2^{(0)}  \bv_2^T \bS^{-1} \bv_1, \nonumber
\\
& a_8 = ({\Halpha}_2^{(0)})^2   \bv_2^T \bS^{-1} \bv_2 + 2 {\Halpha}_2^{(0)}  \bv_2^T \bS^{-1} \bz_1 +  \bz_1^T \bS^{-1} \bz_1  + 1, \nonumber
\\
& a_9 = 2 \bv_2^T \bS^{-1} \bv_2, \ a_{10} = 2 {\Halpha}_2^{(0)}  \bv_1^T \bS^{-1} \bv_2 - 2  \bv_2^T \bS^{-1} \bz_2, \ \nonumber
\\
& a_{11} =  \bv_1^T \bS^{-1} \bv_2, \nonumber
\\
& a_{12} = -\bz_1^T \bS^{-1} \bv_2 -  \bv_1^T \bS^{-1} \bz_2  \nonumber
\\
&+ {\Halpha}_2^{(0)} (\bv_1^T \bS^{-1} \bv_1-\bv_2^T \bS^{-1} \bv_2), \nonumber
\\
& a_{13} = \bz_1^T \bS^{-1} \bz_2  + {\Halpha}_2^{(0)} (\bv_2^T \bS^{-1} \bz_2 - \bz_1^T \bS^{-1} \bv_1)  \nonumber
\\
&-({\Halpha}_2^{(0)})^2  \bz_2^T \bS^{-1} \bv_1, \nonumber
\\
& a_{14} = 2 \bv_1^T \bS^{-1} \bv_2, \ a_{15} = -\bv_1^T \bS^{-1} \bz_2 + {\Halpha}_2^{(0)}  \bv_1^T \bS^{-1} \bv_1 \nonumber
\end{align}
\begin{align}
% \\
& -\bv_2^T \bS^{-1} \bz_1  - {\Halpha}_2^{(0)}  \bv_2^T \bS^{-1} \bv_2.
\end{align}

The calculations to obtain the real coefficients of equation (\ref{eqn:cubic2}) are analogous to those used for the coefficients of (\ref{eqn:thirdDegree})
and lead to
\begin{align}
d_1 &= c_1 c_3 + c_9 c_6 - 2 c_{11} c_{14},
\\
d_2 &= c_1 c_4 + c_2 c_3 + c_6 c_{10} + c_9 c_7 - 2 c_{11} c_{15} - 2 c_{12} c_{14},
\\
d_3 &= c_1 c_5 + c_2 c_4 + c_7 c_{10} + c_9 c_8 - 2 c_{12} c_{15} - 2 c_{13} c_{14},
\\
d_4 &= c_2 c_5 + c_8 c_{10} - 2 c_{13} c_{15},
\end{align}
where
\begin{align}
& c_1 = 2 \bv_2^T \bS^{-1} \bv_2, \ c_2 = -2 {\Halpha}_1^{(1)} \bv_1^T \bS^{-1} \bv_2 + 2 \bv_2^T \bS^{-1} \bz_1, \nonumber
\\
& c_3 = \bv_1^T \bS^{-1} \bv_1, \ c_4 = -2  \bv_1^T \bS^{-1} \bz_2  + 2 {\Halpha}_1^{(1)}  \bv_2^T \bS^{-1} \bv_1, \nonumber
\\
& c_5 = ({\Halpha}_1^{(1)})^2 \bv_2^T \bS^{-1} \bv_2 - 2 {\Halpha}_1^{(1)}  \bv_2^T \bS^{-1} \bz_2  +   \bz_2^T \bS^{-1} \bz_2  + 1, \nonumber
\\
& c_6 = c_1/2, \ c_7 = 2  \bv_2^T \bS^{-1} \bz_1  - 2 {\Halpha}_1^{(1)}  \bv_2^T \bS^{-1} \bv_1, \nonumber
\\
& c_8 = ({\Halpha}_1^{(1)})^2   \bv_1^T \bS^{-1} \bv_1 - 2 {\Halpha}_1^{(1)}  \bv_1^T \bS^{-1} \bz_1 +  \bz_1^T \bS^{-1} \bz_1  + 1, \nonumber
\\ 
& c_9 = 2 \bv_1^T \bS^{-1} \bv_1, \ c_{10} = 2 {\Halpha}_1^{(1)}  \bv_1^T \bS^{-1} \bv_2 - 2  \bv_1^T \bS^{-1} \bz_2, \ \nonumber
\\
& c_{11} = -\bv_1^T \bS^{-1} \bv_2, \ c_{12} = -\bz_1^T \bS^{-1} \bv_1 \nonumber
\\
& +  \bv_2^T \bS^{-1} \bz_2  + {\Halpha}_1^{(1)} (\bv_1^T \bS^{-1} \bv_1-\bv_2^T \bS^{-1} \bv_2), \nonumber
\\
& c_{13} = \bz_1^T \bS^{-1} \bz_2  - {\Halpha}_1^{(1)} (\bv_2^T \bS^{-1} \bz_1 - \bz_2^T \bS^{-1} \bv_1) \nonumber
\\
&+ ({\Halpha}_1^{(1)})^2  \bv_2^T \bS^{-1} \bv_1, \ c_{14} = -2 \bv_1^T \bS^{-1} \bv_2, \nonumber
\\
& \ c_{15} = \bv_2^T \bS^{-1} \bz_2 - {\Halpha}_1^{(1)}  \bv_2^T \bS^{-1} \bv_2
-   \bv_1^T \bS^{-1} \bz_1 + {\Halpha}_1^{(1)}  \bv_1^T \bS^{-1} \bv_1.
\end{align}

\section{Proof of Proposition \ref{proposition:convergence}}
\label{appendix:convergence}
Our goal is to compute a stationary point of
\begin{multline}
h(\alpha_1,\alpha_2)=[1+(\bz_{1} - \bm_{1}(\alpha_1,\alpha_2))^T \bS^{-1} (\bz_{1} - \bm_{1}(\alpha_1,\alpha_2))]
\\
\times [1+(\bz_{2} - \bm_{2}(\alpha_1,\alpha_2))^T \bS^{-1} (\bz_{2} - \bm_{2}(\alpha_1,\alpha_2))] \nonumber
\\
-[(\bz_{1} - \bm_{1}(\alpha_1,\alpha_2))^T \bS^{-1} (\bz_{2} - \bm_{2}(\alpha_1,\alpha_2))]^2,
\end{multline}
which is a continuous, coercive, and strictly 
positive function of $(\alpha_1, \alpha_2)\in\R^2$. It follows that it admits the global minimum and the gradient of $h(\alpha_1,\alpha_2)$
evaluated at this point is zero.

According to Algorithm \ref{algorithm:AlgorithmPP} and assuming that $\Halpha_2^{(0)}$ is known, it is clear that
(by construction) the estimates $\Halpha_i^{(m)}$, $i=1,2$, $m\in \N$, fulfill the following inequality chain:
$h(\Halpha_1^{(1)},\Halpha_2^{(1)}) \geq h(\Halpha_1^{(2)},\Halpha_2^{(2)}) 
\geq \ldots \geq h(\Halpha_1^{(m)},\Halpha_2^{(m)}) \geq h(\Halpha_1^{(m+1)},\Halpha_2^{(m+1)}) \geq \ldots$.
Thus, we have obtained a non-negative (or lower bounded) and decreasing sequence
$g_m = h(\Halpha_1^{(m)},\Halpha_2^{(m)}), \quad m\in \N$, which is thus convergent to a finite limit, i.e.,
\be
\lim_{m\rightarrow +\infty} g_m = \inf_{m\in\N} g_m = g <+\infty.
\label{eqn:limiteSuccessione}
\ee
% It is clear that each value of the sequence $g_n$ is generated by the pairs $a_n=(\Halpha_1^{(l)},\Halpha_2^{(m)})$, $n=l+m$, $l\geq m$, which, in turn,
% form a sequence
% \begin{multline}
% a_1=(\Halpha_1^{(1)},\Halpha_2^{(0)}), \ a_2=(\Halpha_1^{(1)},\Halpha_2^{(1)}), \ldots, \ a_{2n-1}=(\Halpha_1^{(n)},\Halpha_2^{(n-1)}),
% \\
% a_{2n}=(\Halpha_1^{(n)},\Halpha_2^{(n)}), \ a_{2n+1}=(\Halpha_1^{(n+1)},\Halpha_2^{(n)}), \ldots.
% \end{multline}
% As a consequence, we can write that $g_n=h(a_n)$, which is a bounded sequence because $g_1<+\infty$ by construction and $g<+\infty$.
Due to the continuity of $h(\alpha_1,\alpha_2)$ in $\R^2$, it is possible to build a set $\Omega$ defined as follows 
$\Omega =\left\{ a=(\Halpha_1,\Halpha_2)\in\R^2: \ h(\Halpha_1,\Halpha_2)=g \right\}\subseteq \R^2$
and, in addition, by the Weierstrass theorem on bounded sequences \cite{Rudin} and (\ref{eqn:limiteSuccessione}), 
there exists a convergent subsequence of $a_m$
that must have limit belonging to $\Omega$. As a matter of fact, denote by $a_{m_k}$, $k\in\N$, the convergent subsequence extracted from  $a_m$ and
consider the subsequence, $g_{m_k}$ say, of $g_m$ induced by $a_{m_k}$. Now, since $g_m$ is regular, then
every subsequence extracted from it is regular and converges to the same limit.
As a consequence, it is possible to extract from $a_m$ a subsequence $a_{m_k}$, $k\in\N$, converging to $a\in\Omega$.

As a final step, we prove that $a=(\Halpha_1,\Halpha_2)$ is a stationary point for $h(\alpha_1,\alpha_2)$.
To this end, denote by $\Halpha_2^*$ the minimum of $h(\Halpha_1, \alpha_2)$, namely
\be
\Halpha_2^* = \arg\min_{\alpha_2} h(\Halpha_1, \alpha_2)
\label{eqn:minimoFa}
\ee
and observe that, since $h(\alpha_1,\alpha_2)$ is coercive,
% due to equation (\ref{eqn:limiteCoercive}), 
$|\Halpha_2^*|<+\infty$.
Now, recall that $a_{m_k}=(\Halpha_1^{(m_k)},\Halpha_2^{(m_k)})$ is a convergent subsequence extracted from $a_m$ and
apply Algorithm \ref{algorithm:AlgorithmPP} assuming  $(\Halpha_1^{(m_k)},\Halpha_2^*)$ as initial point, then we obtain that
$h(\Halpha_1^{(m_k)},\Halpha_2^*)\geq h(\Halpha_1^{(m_k)},\Halpha_2^{(m_k)}) \geq h(\Halpha_1^{(m_k+1)},\Halpha_2^{(m_k+1)})$.
Since $g_n$ is decreasing, it follows that $g_{m_k+1}\geq g_{m_{k+1}}$ and, hence, the above chain of inequalities continues as follows
$h(\Halpha_1^{(m_k)},\Halpha_2^*)\geq h(\Halpha_1^{(m_k+1)},\Halpha_2^{(m_k+1)}) \geq h(\Halpha_1^{(m_{k+1})},\Halpha_2^{(m_{k+1})})$,
which, for $k\rightarrow +\infty$, becomes
$h(\Halpha_1,\Halpha_2^*)\geq h(\Halpha_1,\Halpha_2)$.
Exploiting (\ref{eqn:minimoFa}) in conjunction with the last equation yields
$h(\Halpha_1, \alpha_2) \geq h(\Halpha_1,\Halpha_2), \quad \forall \alpha_2 \in \R$.
Moreover, since $h(\alpha_1,\alpha_2)$ is continuous and differentiable with respect to $\alpha_2$, it is clear that
\begin{multline}
\left[ \frac{\partial}{\partial \alpha_2} h(\Halpha_1,\alpha_2) \right]_{\!\alpha_2=\Halpha_2} \!\!\!\!\!=0 
\Rightarrow 
\left[ \frac{\partial}{\partial \alpha_2} h(\alpha_1,\alpha_2) \right]_{\!\!\!\!\scriptsize \begin{array}{l}
									 \alpha_1=\Halpha_1
									  \\
									  \vspace{-3mm}
									  \\
                                                                          \alpha_2=\Halpha_2
                                                                        \end{array}}\!\!\!\!\!\!\!=0.
                                                                        \label{eqn:partAlpha1}
\end{multline}
Following the same line of reasoning, it is not difficult to show that
\begin{multline}
\left[ \frac{\partial}{\partial \alpha_1} h(\alpha_1,\Halpha_2) \right]_{\! \alpha_1=\Halpha_1}\!\!\!\!\!\!\!=0
\Rightarrow
\left[ \frac{\partial}{\partial \alpha_1} h(\alpha_1,\alpha_2) \right]_{\!\!\!\!\scriptsize \begin{array}{l}
									\alpha_1=\Halpha_1
									  \\
									  \vspace{-3mm}
									  \\
                                                                         \alpha_2=\Halpha_2
                                                                        \end{array}}\!\!\!\!\!\!=0.
                                                                        \label{eqn:partAlpha2}
\end{multline}
Notice that equations (\ref{eqn:partAlpha1}) and (\ref{eqn:partAlpha2}) can be written in a more compact form using the gradient operator
\be
\nabla h(\Halpha_1,\Halpha_2)=0,
\ee
which implies that $a=(\Halpha_1,\Halpha_2)$ is a stationary point of $h(\alpha_1,\alpha_2)$.

\section{Derivation of the Rao test}
\label{appendix:derivationRao}
As first step, we denote by $s(\alpha_1,\alpha_2,\bM)$ the natural logarithm of the pdf of $\bZ$ and $\bZ_s$ under $H_1$,
namely
\begin{multline}
s(\alpha_1,\alpha_2,\bM)=-N\ln(2\pi)-\ln[\det(\bM)]
\\
-\frac{1}{2}(\bz_1 - \bm_1(\alpha_1,\alpha_2))^T\bM^{-1}(\bz_1 - \bm_1(\alpha_1,\alpha_2))
\\
-\frac{1}{2}(\bz_2 - \bm_2(\alpha_1,\alpha_2))^T\bM^{-1}(\bz_2 - \bm_2(\alpha_1,\alpha_2))
\\
+\ln[f(\bZ_S;\bM)].
\end{multline}
It is not difficult to show that
\begin{align}
& \frac{\partial}{\partial \alpha_1}s(\alpha_1,\alpha_2,\bM)=
\bv_1^T\bM^{-1}( \bz_1 - \alpha_1 \bv_1 + \alpha_2\bv_2) \nonumber
\\
&+\bv_2^T\bM^{-1}( \bz_2 - \alpha_1 \bv_2 - \alpha_2\bv_1),
\label{eqn:partialA1}
\end{align}
\begin{align}
% \\
& \frac{\partial}{\partial \alpha_2}s(\alpha_1,\alpha_2,\bM)=
-\bv_2^T\bM^{-1}( \bz_1 - \alpha_1 \bv_1 + \alpha_2\bv_2) \nonumber
\\
&+\bv_1^T\bM^{-1}( \bz_2 - \alpha_1 \bv_2 - \alpha_2\bv_1).
\label{eqn:partialA2}
\end{align}
Thus, we can write
\begin{multline}
\left\{\frac{\partial}{\partial \btheta_A}\big[\ln\big(f(\bZ;\bM,\alpha_1,\alpha_2,H_1)f(\bZ_S;\bM)\big)\big]\right\}_{\btheta=\btheta_0}
\\
=\left[
\begin{array}{c}
\bv_1^T\bS_0^{-1}\bz_1 + \bv_2^T\bS_0^{-1}\bz_2
\\
-\bv_2^T\bS_0^{-1}\bz_1+\bv_1^T\bS_0^{-1}\bz_2
\end{array}
\right].
\label{eqn:scoreDer}
\end{multline}
Now, let us focus on the Fisher information matrix and exploit (\ref{eqn:partialA1}) and (\ref{eqn:partialA2}) to evaluate
\begin{multline}
\bF_{AA}(\btheta)=-E\left[
\frac{\partial^2}{\partial \btheta_A \partial \btheta^T_A}s(\alpha_1,\alpha_2,\bM)
\right]
\\
=\left[
\begin{array}{cc}
\bv_1^T\bM^{-1}\bv_1 + \bv_2^T\bM^{-1}\bv_2 & \bzero
\\
\bzero & \bv_1^T\bM^{-1}\bv_1 + \bv_2^T\bM^{-1}\bv_2
\end{array}
\right].
\label{eqn:FAA}
\end{multline}
As a final step towards the derivation of the Rao test, we only need to notice that
\begin{align}
&\bF_{AB}(\btheta)=-E\left[
\frac{\partial^2}{\partial \btheta_A \partial \btheta^T_B}[s(\alpha_1,\alpha_2,\bM)]
\right]
\\
&=-E\left[
\begin{array}{c}
\ds\frac{\partial}{\partial \bff^T(\bM)} [\bv_1^T\bM^{-1}(\bz_1-\bm_1(\alpha_1,\alpha_2))
\\
\vspace{-3mm}
\\
\ds\frac{\partial}{\partial \bff^T(\bM)} [-\bv_2^T\bM^{-1}(\bz_1-\bm_1(\alpha_1,\alpha_2))
\end{array}
\right.\nonumber
\\
&\left.
\begin{array}{c}
\ds \phantom{\ds\frac{\partial}{\partial \bff^T(\bM)}} +\bv_2^T\bM^{-1}(\bz_2-\bm_2(\alpha_1,\alpha_2))]
\\
\vspace{-3mm}
\\
\ds \phantom{\ds\frac{\partial}{\partial \bff^T(\bM)}} +\bv_1^T\bM^{-1}(\bz_2-\bm_2(\alpha_1,\alpha_2))]
\end{array}
\right]
% \end{align}
% \begin{align}
\\
&=-E\left[
\begin{array}{c}
\ds \bff^T\left\{\frac{\partial}{\partial \bM} [\bv_1^T\bM^{-1}(\bz_1-\bm_1(\alpha_1,\alpha_2))
\right.
\\
\vspace{-3mm}
\\
\ds\bff^T\left\{\frac{\partial}{\partial \bM} [-\bv_2^T\bM^{-1}(\bz_1-\bm_1(\alpha_1,\alpha_2))
\right.
\end{array}
\right.\nonumber
\\
& \left. 
\begin{array}{c}
\ds \left. \phantom{\frac{\partial}{\partial \bM}} +\bv_2^T\bM^{-1}(\bz_2-\bm_2(\alpha_1,\alpha_2))]
\right\}
\\
\vspace{-3mm}
\\
\ds \left. \phantom{\frac{\partial}{\partial \bM}} +\bv_1^T\bM^{-1}(\bz_2-\bm_2(\alpha_1,\alpha_2))]
\right\}
\end{array}
\right]
\\
&=-\left[\begin{array}{c}
\bff^T \{E[-\bM^{-1}\bv_1(\bz_1-\bm_1(\alpha_1,\alpha_2))^T\bM^{-1} 
\\
\bff^T\{E[\bM^{-1}\bv_2(\bz_1-\bm_1(\alpha_1,\alpha_2))^T\bM^{-1} 
\end{array}
\right. \nonumber
\\
&\left.\begin{array}{c}
-\bM^{-1}\bv_2(\bz_2-\bm_2(\alpha_1,\alpha_2))^T\bM^{-1}] \}
\\
- \bM^{-1}\bv_1(\bz_2-\bm_2(\alpha_1,\alpha_2))^T\bM^{-1}] \}
\end{array}
\right]
\\
&=\bzero.
\end{align}
As a consequence, (\ref{eqn:FAAmeno1}) simplifies as
\be
[\bF(\btheta)^{-1}]_{AA} = [\bF_{AA}(\btheta)]^{-1}
\ee
and, hence,
\begin{multline}
[\bF(\btheta_0)^{-1}]_{AA}
\\
=\left[
\begin{array}{cc}
\ds\frac{1}{\bv_1^T\bS_0^{-1}\bv_1 + \bv_2^T\bS_0^{-1}\bv_2} & \ds\bzero
\\
\ds\bzero & \ds\frac{1}{\bv_1^T\bS_0^{-1}\bv_1 + \bv_2^T\bS_0^{-1}\bv_2}
\end{array}
\right].
\end{multline}
Using the above equation in conjunction with (\ref{eqn:scoreDer}) leads to
the final expression of the Rao test statistic
\be
\tRAO = \frac{t_1(\bS_0)+t_2(\bS_0)}
{\bv_1^T\bS_0^{-1}\bv_1 + \bv_2^T\bS_0^{-1}\bv_2}.
\ee

% 
% \vfill
% \pagebreak

%\newpage
%
%
\begin{figure}[ht!]
\begin{center}
\includegraphics[scale=0.45]{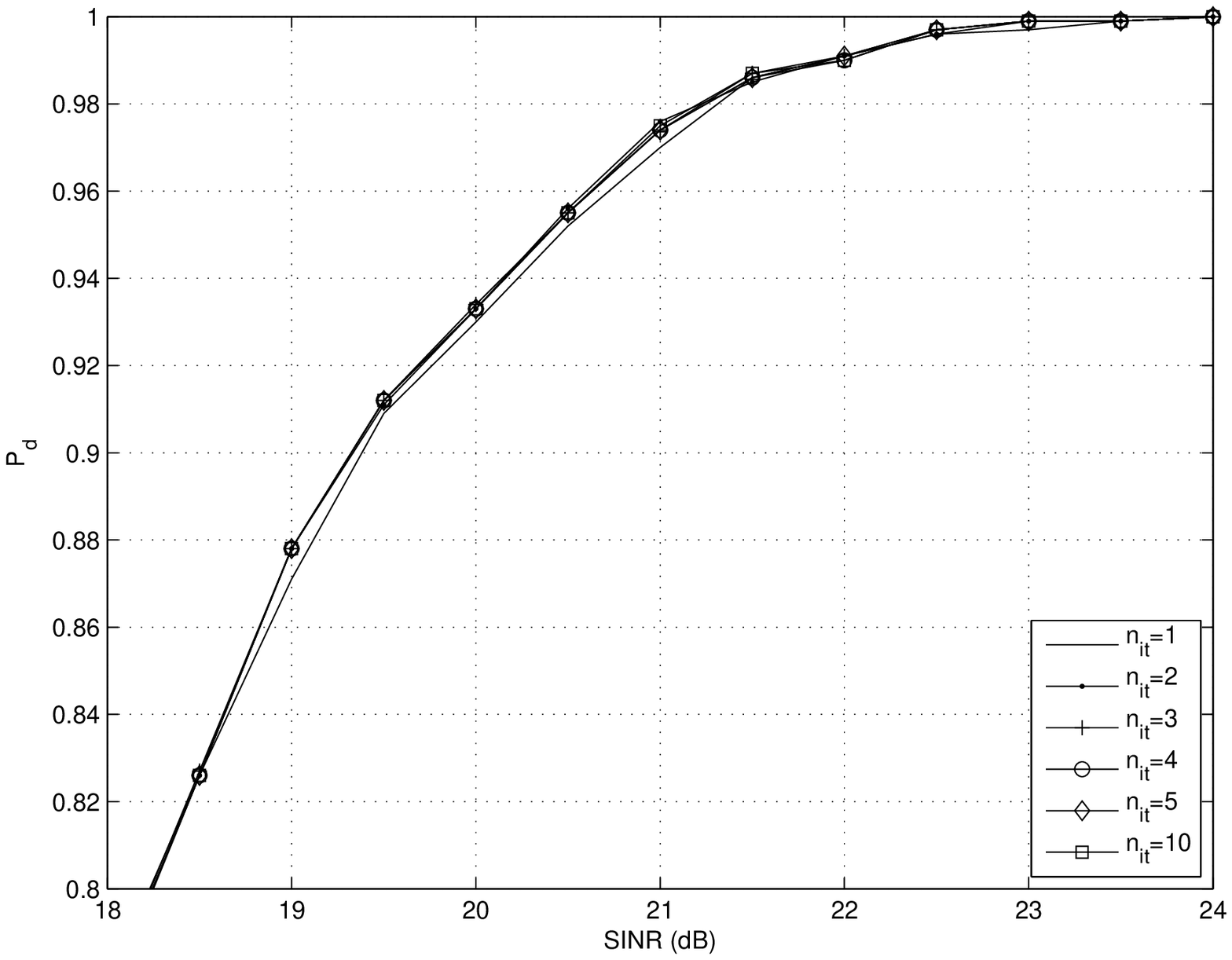}
\caption{$P_d$ versus SINR for the I-GLRT assuming $N=8$, $K=6$, $\nu_d=0.1$, and $n$ as parameter.}
\label{fig:figure01}
\end{center}
\end{figure}
%
% \clearpage
%
\begin{figure}[ht!]
\begin{center}
\includegraphics[scale=0.45]{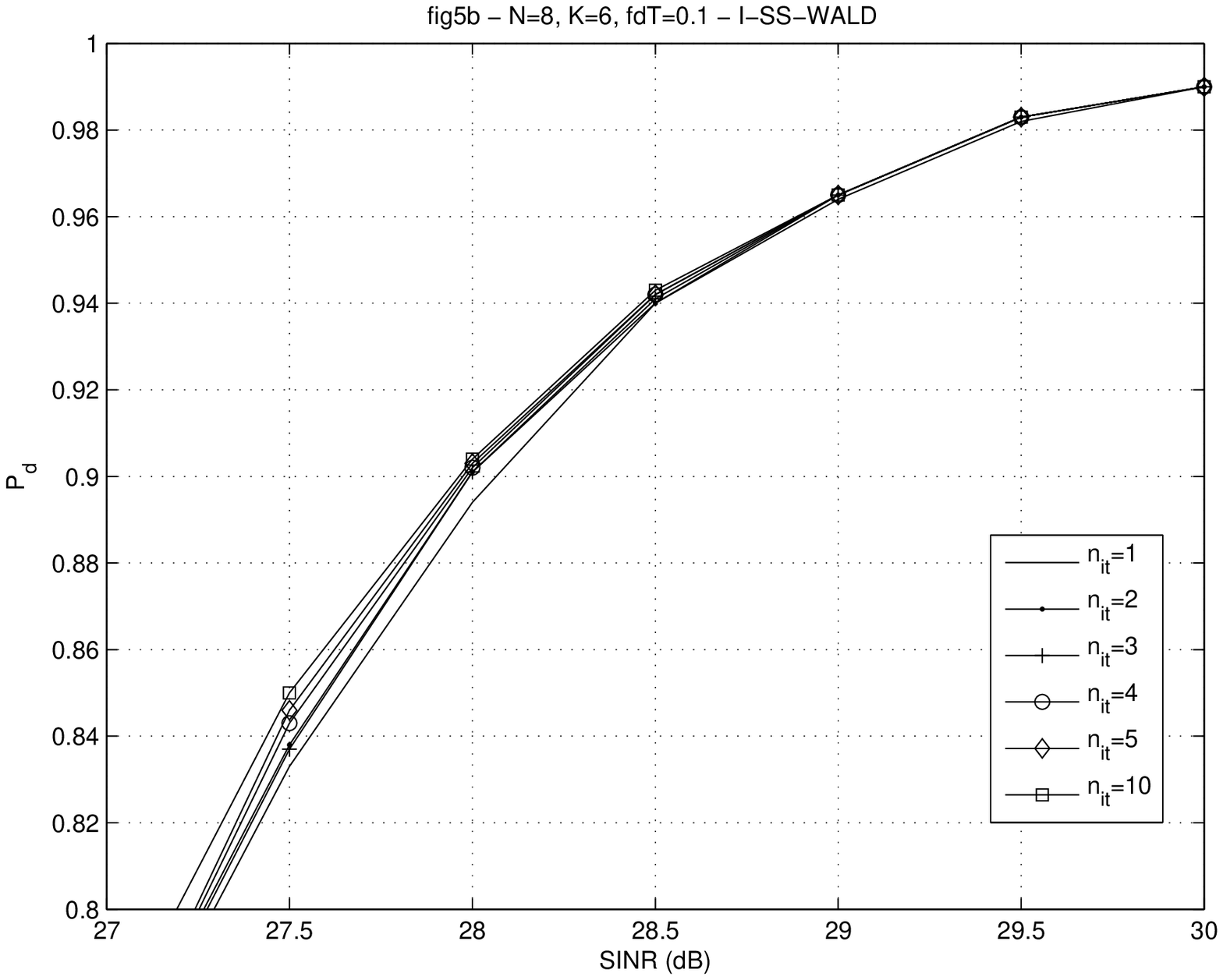}
\caption{$P_d$ versus SINR for the I-WALD assuming $N=8$, $K=6$, $\nu_d=0.1$, and $n$ as parameter.}
\label{fig:figure02}
\end{center}
\end{figure}
\begin{figure}[ht!]
\begin{center}
\includegraphics[width=9cm,height=6cm]{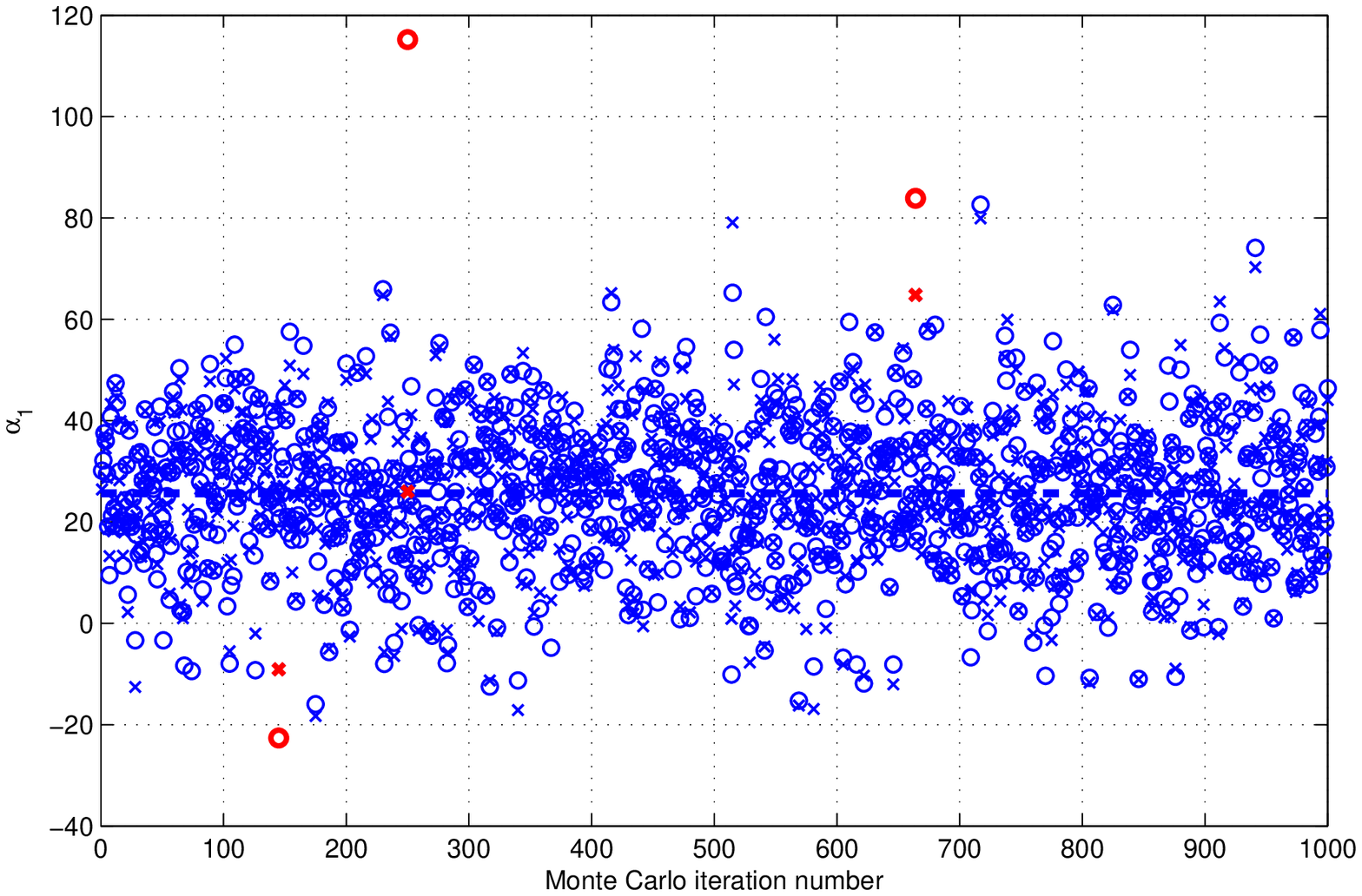}
\caption{Estimate of $\alpha_1$ versus Monte Carlo iteration number; estimates provided by
Algorithm 1 (cross marker and no line) and the two-step design procedure (circle marker and no line). 
In addition, the actual value of $\alpha_1$ is also plotted (no marker and dashed line).}
\label{fig:figure03}
\end{center}
\end{figure}
\begin{figure}[ht!]
\begin{center}
\includegraphics[width=9cm,height=6cm]{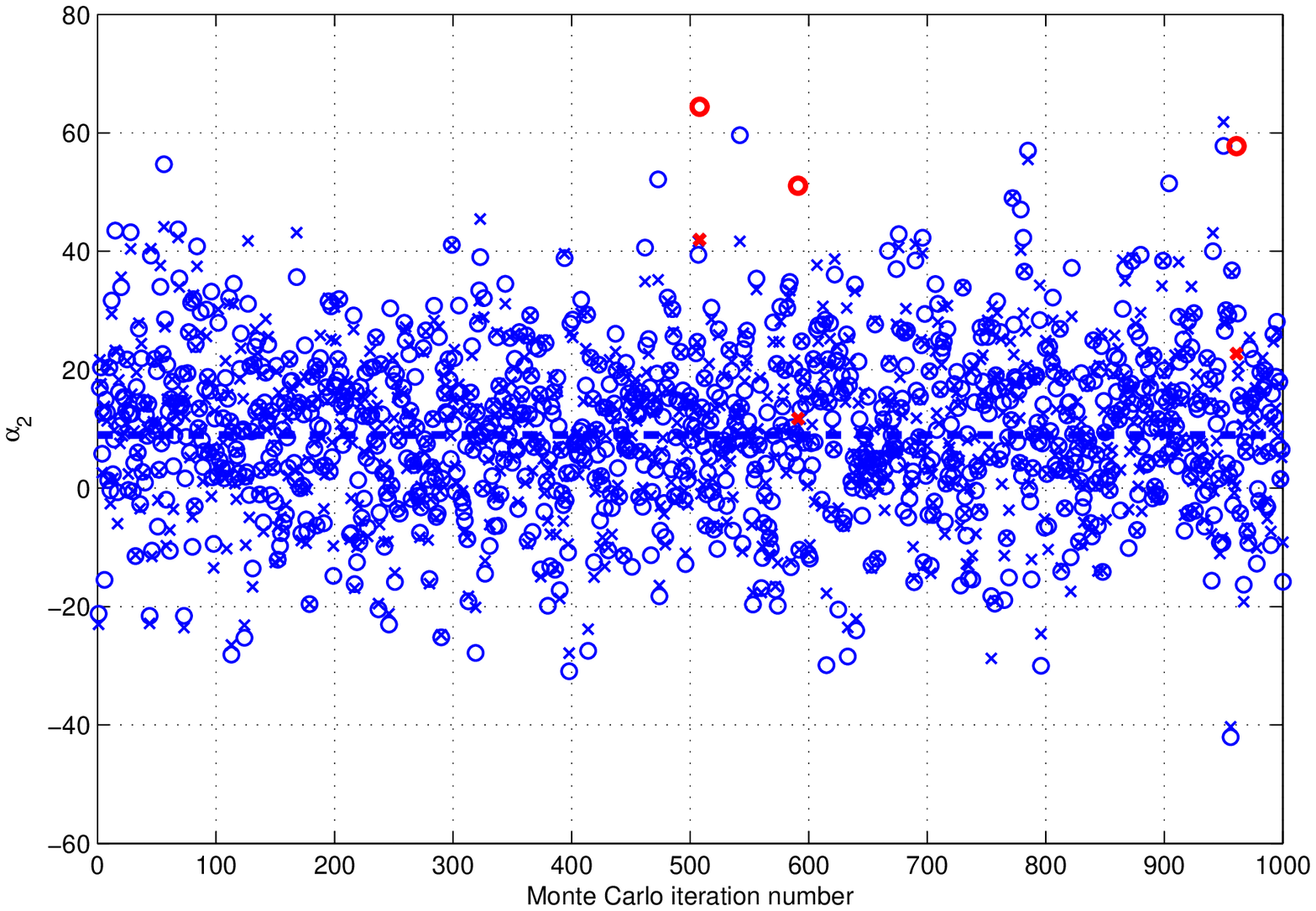}
\caption{Estimate of $\alpha_2$ versus Monte Carlo iteration number; estimates provided by
Algorithm 1 (cross marker and no line) and the two-step design procedure (circle marker and no line). 
In addition, the actual value of $\alpha_2$ is also plotted (no marker and dashed line).}
\label{fig:figure03a}
\end{center}
\end{figure}
\begin{figure}[ht!]
\begin{center}
\includegraphics[scale=0.45]{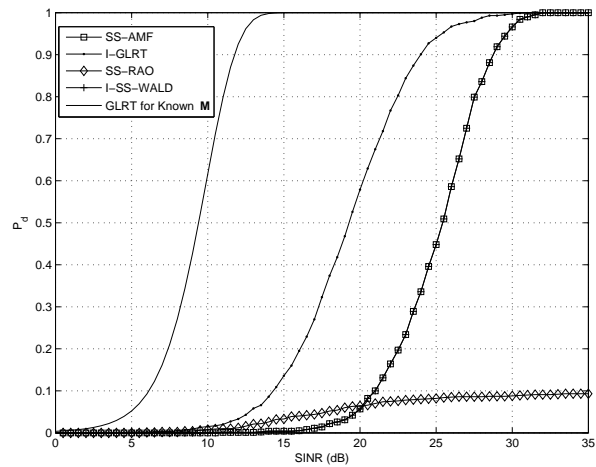}
\caption{$P_d$ versus SINR for the SS-AMF, the I-GLRT, the SS-RAO, and the I-WALD assuming $N=8$, $K=6$, $\nu_d=0$, and $n=3$.}
\label{fig:figure04}
\end{center}
\end{figure}
\begin{figure}[ht!]
\begin{center}
\includegraphics[scale=0.45]{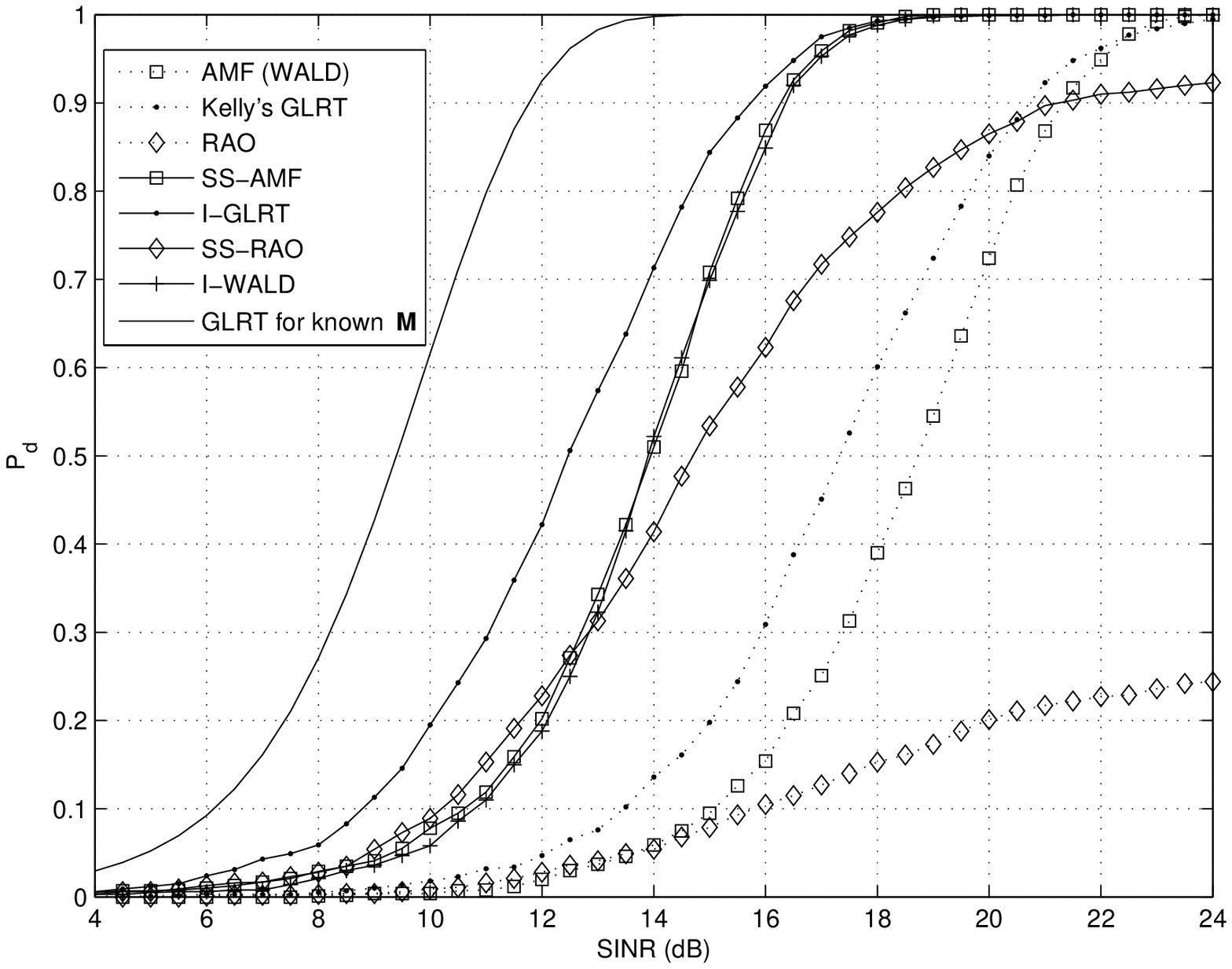}
\caption{$P_d$ versus SINR for the SS-AMF, the I-GLRT, the SS-RAO, the I-WALD, Kelly's GLRT, the AMF, and the RAO assuming $N=8$, $K=12$, $\nu_d=0$, and $n=3$.}
\label{fig:figure05}
\end{center}
\end{figure}
\begin{figure}[ht!]
\begin{center}
\includegraphics[scale=0.45]{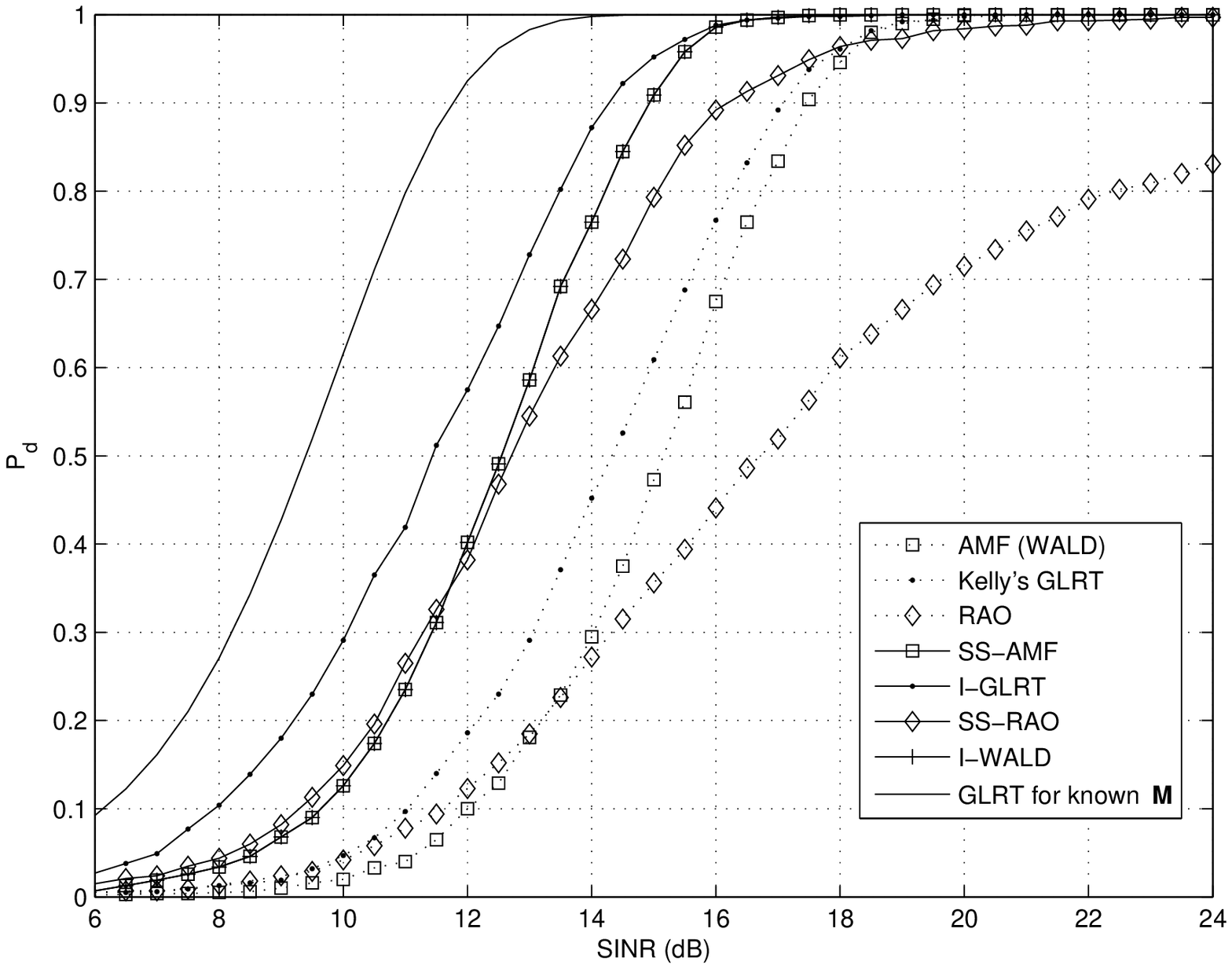}
\caption{$P_d$ versus SINR for the SS-AMF, the I-GLRT, the SS-RAO, the I-WALD, Kelly's GLRT, the AMF, and the RAO assuming $N=8$, $K=16$, $\nu_d=0$, and $n=3$.}
\label{fig:figure06}
\end{center}
\end{figure}
\begin{figure}[ht!]
\begin{center}
\includegraphics[scale=0.45]{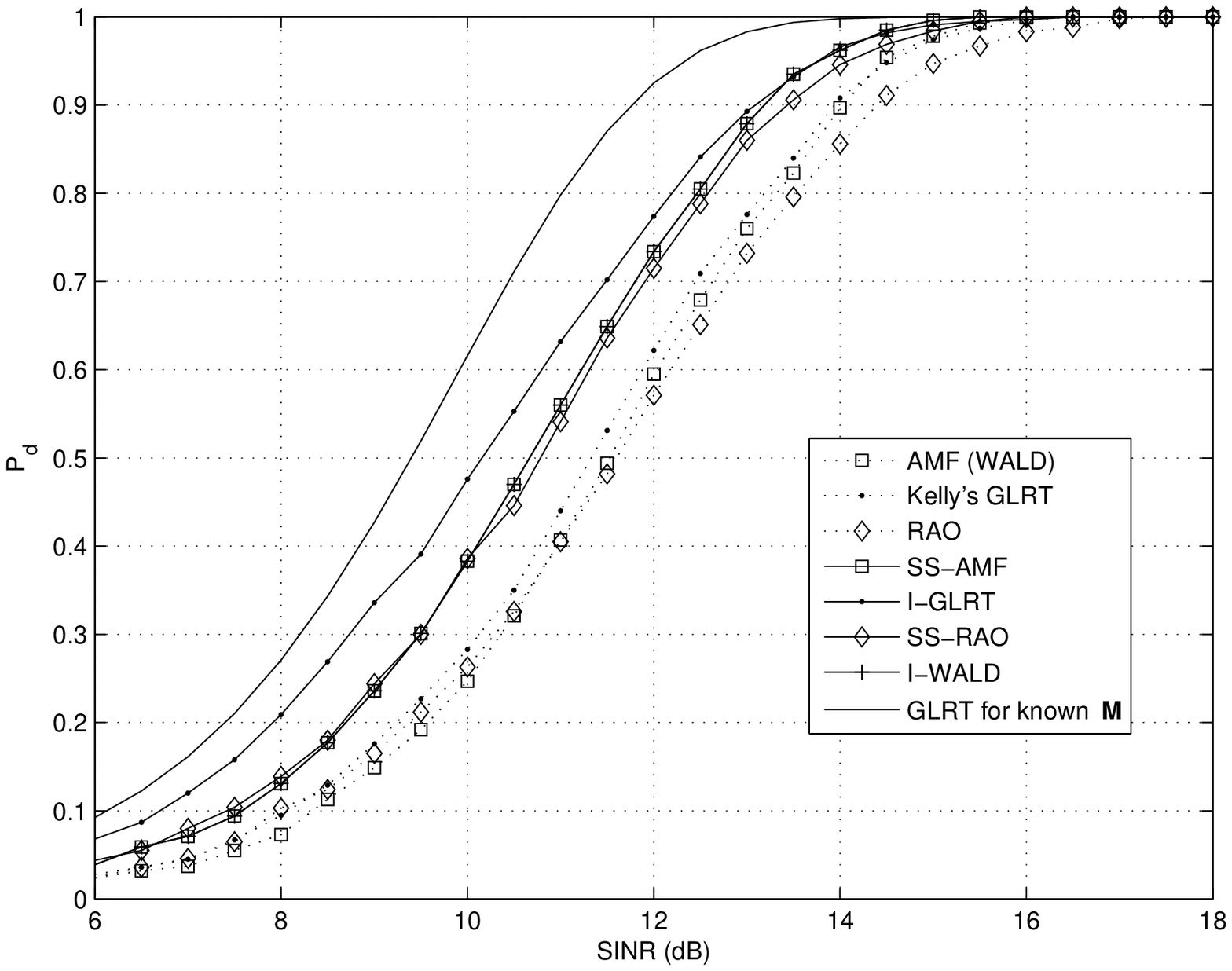}
\caption{$P_d$ versus SINR for the SS-AMF, the I-GLRT, the SS-RAO, the I-WALD, Kelly's GLRT, the AMF, and the RAO assuming $N=8$, $K=32$, $\nu_d=0$, and $n=3$.}
\label{fig:figure07}
\end{center}
\end{figure}
\begin{figure}
\centering
\subfigure[HH Polarization]
{\includegraphics[scale=0.45]{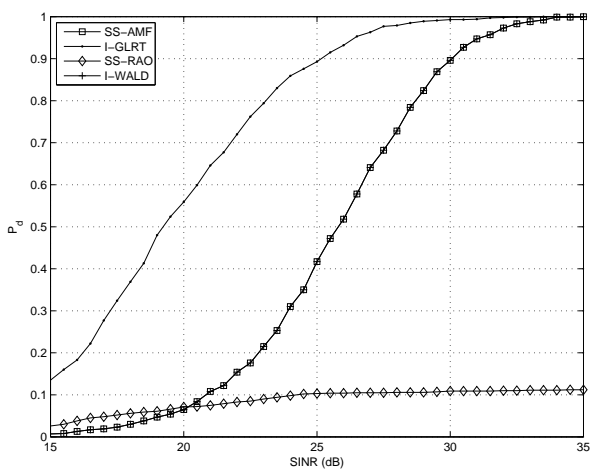}}%\includegraphics[width=8cm,height=7cm]{figure09a.ps}}
\hspace{5mm}
\subfigure[VV Polarization]
{\includegraphics[scale=0.45]{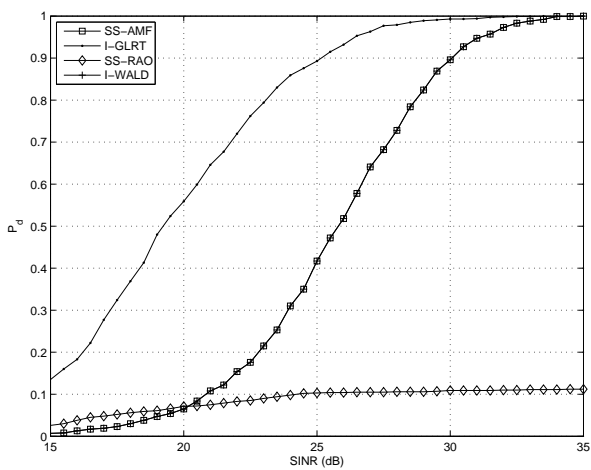}}%\includegraphics[width=8cm,height=7cm]{figure09b.ps}}
\caption{$P_d$ versus SINR for the SS-AMF, the I-GLRT, the SS-RAO, and the I-WALD assuming $N=8$, $K=6$, $\nu_d=0$, $n=3$;
performance are evaluated on MIT-LL Phase-One radar dataset.}
\label{fig:figure09}
\end{figure}
\begin{figure}[ht!]
\begin{center}
\includegraphics[scale=0.45]{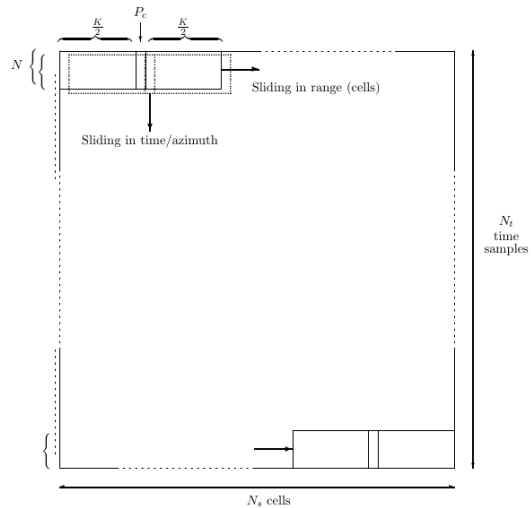}
\caption{Data selection procedure for the CFAR analysis.}
\label{fig:figure08}
\end{center}
\end{figure}
\begin{figure}
\centering
\subfigure[HH Polarization]
{\includegraphics[scale=0.45]{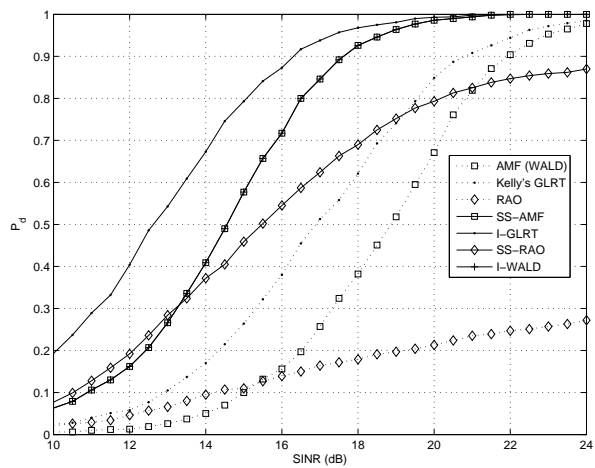}}%\includegraphics[width=8cm,height=7cm]{figure10a.ps}}
\hspace{5mm}
\subfigure[VV Polarization]
{\includegraphics[scale=0.45]{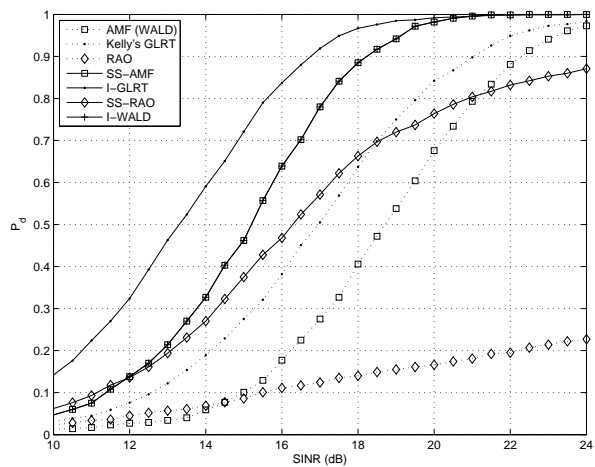}}%\includegraphics[width=8cm,height=7cm]{figure10b.ps}}
\caption{$P_d$ versus SINR for the SS-AMF, the I-GLRT, the SS-RAO, the I-WALD, Kelly's GLRT, the AMF, and the RAO assuming $N=8$, $K=12$, $\nu_d=0$, $n=3$;
performance are evaluated on MIT-LL Phase-One radar dataset.}
\label{fig:figure10}
\end{figure}
\end{document}